\documentclass[journal,onecolumn,12pt]{IEEEtran}

\usepackage{calligra} 
\usepackage[T1]{fontenc} 
\usepackage{enumerate} 
\usepackage[shortlabels]{enumitem}
\usepackage{amssymb}
\usepackage{amsmath,mathtools}
\usepackage{bm} 
\usepackage{amsthm}
\usepackage{cite}
\newcommand{\D}{\mathsf{D}}
\newcommand{\dd}{\text{d}}

\usepackage{pifont} 
\usepackage{tikz,graphicx,pgfplots}
\usetikzlibrary{calc,positioning,spy}
\renewcommand{\Re}{\mbox{Re}}
\renewcommand{\Im}{\mbox{Im}}

\ifCLASSOPTIONcompsoc
\usepackage[caption=false,font=normalsize,labelfon
t=sf,textfont=sf]{subfig}
\else
\usepackage[caption=false,font=footnotesize]{subfig}
\fi

\newtheorem{theorem}{Theorem}

\theoremstyle{definition} 

\title{Maximum Likelihood Detection in a Four-Dimensional
Stokes-Space Receiver}
\author{Amir Tasbihi~\IEEEmembership{Graduate Student Member, IEEE} and Frank R. Kschischang~\IEEEmembership{Fellow, IEEE}\thanks{The authors are with the
Edward S. Rogers Sr.\ Dept.\ of Electrical and Computer Engineering,
University of Toronto, Toronto, ON M5S 3G4, Canada.  Email:
\texttt{\{tasbihi,frank\}@ece.utoronto.ca}.  Submitted to
\emph{IEEE Trans. Commun.}, February 1, 2019.}}

\begin{document}
\maketitle
\begin{abstract}
The maximum likelihood detection rule for a four dimensional direct-detection optical front-end is derived. The four dimensions are two intensities and two differential phases. Three different signal processing algorithms, composed of symbol-by-symbol, sequence and successive detection, are discussed. To remedy dealing with special functions in the detection rules, an approximation for high signal-to-noise ratios (SNRs) is provided. Simulation results show that, despite the simpler structure of the successive algorithm, the resulting performance loss, in comparison with the other two algorithms, is negligible. For example, for an 8-ring/8-ary phase constellation, the complexity of detection reduces by a factor of 8, while the performance, in terms of the symbol error rate, degrades by 0.5~dB. It is shown that the high-SNR approximation is very accurate, even at low SNRs. The achievable rates for different constellations are computed and compared by the Monte Carlo method. For example, for a 4-ring/8-ary phase constellation, the achievable rate is 10 bits per channel use at an SNR of 25~dB, while by using an 8-ring/8-ary phase constellation and an error correcting code of rate 5/6, this rate is achieved at an SNR of 20~dB.   

   
\end{abstract}

\clearpage
\begin{section}{Introduction}
\label{sec:introduction}
\IEEEPARstart{W}e derive the maximum likelihood (ML) detection rule for the four dimensional direct-detection receiver optical front-end described in~\cite{osman}. We propose three digital signal processing (DSP) algorithms: symbol-by-symbol, sequence (min-sum) and successive detection. For each proposed method, we find the likelihood (utility) function to be maximized for the ML detection.

Due to their inexpensive structures, non-coherent detection schemes have promising applications in short-haul ($<100$~km) data transmission, {\it e.g.,} intra data-center communication~\cite{future,6887234}. In addition, the demand for high data rates necessitates the usage of all degrees of freedom (DOF) for data transmission. 

Although exploiting the intensity or the phase of the transmitted light without using any local oscillator at the receiver was proposed earlier, exploiting both of them simultaneously in optical communication goes back to the early 2000s~\cite{1159095}. To combat the practical issue of precise adjustment in those techniques, {\it self-homodyne} detection was proposed in 2005~\cite{1432820}, which later was improved to exploit both of the polarizations for the data transmission in wavelength division multiplexing systems~\cite{5621405,5713203}. Similar to self-homodyne detection, {\it Stokes-vector direct detection} (SVDD) was introduced in 2014, taking into account the polarization rotation of the fiber~\cite{6887234,6936845}. Despite the simple structure of the receiver, SVDD (and also self-homodyne detection) devotes half of the available dimensions to the transmission of a pilot symbol. Later, a modified Stokes-space direct detection scheme was introduced in~\cite{osman} which, by transmitting a data symbol instead of the pilot, achieves a higher data-rate. However, exploiting all of its DOF is only possible under either non-realistic assumptions or by using a complex receiver. This issue was later resolved for the same optical front-end by additional processing in the DSP~\cite{tasbihi}. The present paper extends the results of~\cite{tasbihi} by determining the actual ML detection rule for the various schemes under study.

\IEEEpubidadjcol
The rest of the paper is organized as follows. The system model, including the transmitter, the channel and the receiver, is introduced in Sec.~\ref{sec:TxChRx}. In Sec.~\ref{sec:JointDetection} we discuss symbol-by-symbol ML detection. In Sec.~\ref{sec:MaxProd} we describe sequence detection, exploiting the min-sum algorithm on the factor graph of the system. To combat the complexity of symbol-by-symbol and sequence detectors, we propose a successive detection scheme in Sec.~\ref{sec:SuccessiveDetection}. Due to the existence of modified Bessel functions in the likelihood scores, we introduce an accurate and easy-to-compute approximation, suitable for operation in the moderate to large SNR regimes. These approximately ML decoders are described at the end of Secs.~\ref{sec:JointDetection}, \ref{sec:MaxProd}, and \ref{sec:SuccessiveDetection}. In Sec.~\ref{sec:fading}, we discuss about the fast fading behaviour of the fourth DOF subchannel. In Sec.~\ref{sec:Results}, we compare the discussed methods in previous sections via simulations. Finally, we provide concluding remarks in Sec.~\ref{sec:Conclusion}.


Through this paper, we adopt the following notational conventions.
\begin{itemize}  
\item Scalars: lower-case letters, {\it e.g.}, $a$ and $\rho$.
\item $|a|,\arg(a),\Re(a),\Im(a),a^\ast$: the magnitude, phase, real and imaginary parts and complex conjugate of the complex number $a$, respectively.
\item Sets: blackboard bold capital letters, {\it e.g.}, $\mathbb{A, B}$. In particular, $\mathbb{Z}$, $\mathbb{R}$, $\mathbb{R}^+$ and $\mathbb{C}$ denote integers, real numbers, non-negative reals, and complex numbers, respectively.
\item $\mathbb{A}$\ding{214}$\mathbb{B}$ : there is a bijection between $\mathbb{A}$ and $\mathbb{B}$.
\item Vectors: lower-case bold letters, {\it e.g.}, $\bm{v}$. The $i^{th}$ element of $\bm{v}$ is denoted by $\bm{v}(i)$. For $\bm{u},\bm{v}\in\mathbb{C}^n$, the inner product is defined as
$\langle\bm{u},\bm{v}\rangle\triangleq\sum_{i=1}^{n}\bm{u}(i)\bm{v}(i)^\ast.$
\item Matrices: upper-case bold letters, {\it e.g.}, $\bm{M}$. In addition, $\bm{I}_{n\times n}$ denotes the $n\times n$ identity matrix.
\item $|\bm{M}|$ and $\bm{M}^t$: the determinant and the transpose of $\bm{M}$.
\item Random variables: non-bold capital letters, {\it e.g.} $A$. Realizations are shown in the same lower-case letter, {\it e.g.,} if $A$ is a random variable then $a$ is its realization.
\item $a_{i:j}$: the sequence $a_i,a_{i+1},\ldots,a_{j-1},a_j$, where $i\leq j$.
\item $\D$: the unit-delay operator, {\it e.g.,} if $x[n]$ denotes a discrete-time signal then $\D x[n]=x[n-1]$.
\item We will extensively make use of the {\it Jones vector} representation of light~\cite{agrawal2012fiber}. For the Jones vector $\bm{v}$, $v_x$ and $v_y$ denote the X and Y polarizations, respectively.
\end{itemize}

\end{section}
\begin{section}{The System Model}
\label{sec:TxChRx}
In this section, we formulate the signal processing operations performed by the transmitter, the channel, and the receiver.
\begin{subsection}{The Transmitter}
\label{subsec:Transmitter} 
For simplicity, we discuss the base-band equivalent model. We assume the use of Nyquist pulses, {\it i.e.,} pulses without intersymbol interference, allowing for a discrete-time formulation corresponding to the sample times. We use an $n_r$-ring/$n_p$-ary phase constellation (See Fig.~\ref{fig:cons}) with equally-spaced squared radii as in~\cite{osman}. 
The radius set is 
\[\left\{r_1,r_1\sqrt{1+\delta^2},r_1\sqrt{1+2\delta^2},\ldots,r_1\sqrt{1+(n_r-1)\delta^2}\right\},\]
where $r_1$ and $\delta\in\mathbb{R}^{+}$, and the phase set is
\[\left\{0,\frac{2\pi}{n_p},\frac{4\pi}{n_p},\ldots,\frac{(n_p-1)2\pi}{n_p}\right\}.\]
The transmitter sends two points, $e_x$ and $e_y$, from the constellation over the X and Y polarizations, respectively. As a result, the transmitted symbol is $\bm{e}=[e_x, e_y]^t\in\mathbb{C}^2$. As $e_x$ and $e_y$ are complex numbers, they have a magnitude and a phase, providing four DOF to exploit. As with any non-coherent scheme, instead of the absolute phase, we use differential phase encoding. As a result, we encode our data in
\begin{enumerate}[i)]
\item $|e_x|$,
\item $|e_y|$,
\item $\theta\triangleq\arg(e_xe_y^\ast)$,
\item $\gamma\triangleq\arg(e_x\cdot\D e_y^\ast)$,
\end{enumerate}
which we refer to as the first up to the fourth dimension, respectively. The relationship among these dimensions are shown in Fig.~\ref{fig:phaseRelationE}. The fourth dimension necessitates an initial condition on a symbol block, which is achieved by transmitting a pilot symbol, {\it e.g.,} $\bm{e}_{\text{pilot}}=[r_1,r_1]^t$, at the beginning of the block. 
\begin{figure}
\centering
\subfloat[]{\includegraphics[scale=2]{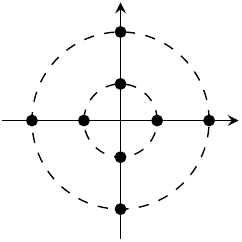}}
\subfloat[]{\includegraphics[scale=2]{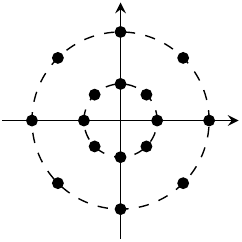}}
\subfloat[]{\includegraphics[scale=2]{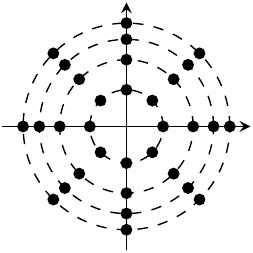}}
\caption{$n_r$-ring/$n_p$-ary phase constellations: (a) $2$-ring/$4$-ary, (b) $2$-ring/$8$-ary, (c) $4$-ring/$8$-ary.}
\label{fig:cons}
\end{figure}
\begin{figure}
\centering
\subfloat[\label{fig:phaseRelationE}]{%
       \includegraphics[scale=1]{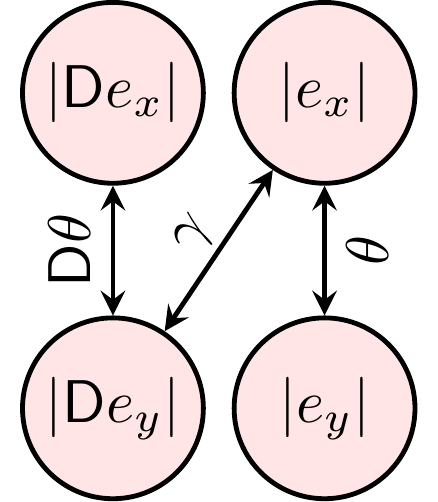}}
\subfloat[\label{fig:phaseRelationK}]{%
       \includegraphics[scale=1]{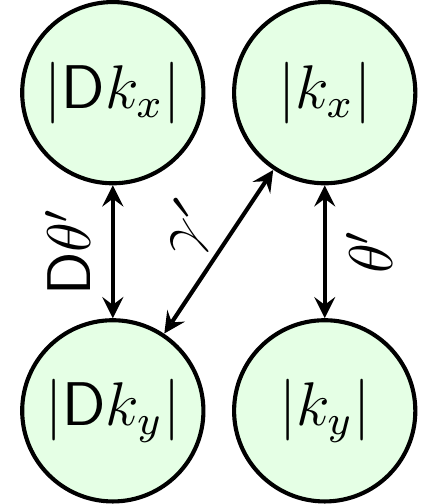}}
\subfloat[\label{fig:phaseRelationR}]{%
       \includegraphics[scale=1]{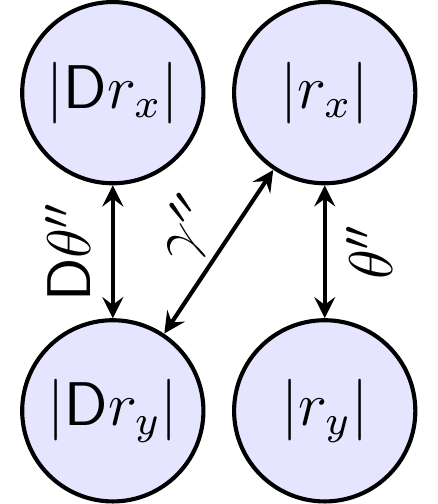}}
\caption{The four dimensions (a) at the transmitter, (b) before the amplifier, and (c) after the amplifier.}
\end{figure}

\end{subsection}

\begin{subsection}{The Channel}
\label{subsec:Channel}
We adopt a linear channel model, neglecting fiber nonlinearities, as is appropriate for short-haul data transmission. The random birefringence of single mode fibers impose a linear transformation on the input Jones vector. In particular, in the absence of noise, the output Jones vector, $\bm{k}$, can be written as
\begin{equation}\bm{k}=\bm{H}\bm{e},
\label{eq:ChannelRotation}
\end{equation}
where 
\begin{equation*}
\bm{H}=\left[\begin{array}{cc}
a & b\\
-b^\ast & a^\ast 
\end{array}\right]
\end{equation*}
is the channel rotation matrix, such that $a,b\in\mathbb{C}$ and $|a|^2+|b|^2=1$~\cite{PMDcomp,gordon2000pmd,kikuchi,osman}. The matrix $\bm{H}$ is nonsingular; so if $\mathbb{E}$ and $\mathbb{K}$ denote all possible $\bm{e}$ and $\bm{k}$ vectors, respectively, then $\mathbb{E}$\ding{214}$\mathbb{K}$. The {\it coherence time} of the channel matrix is assumed to be much larger than a symbol duration, so we can neglect its variation over the transmission of a sequence of symbols. Fiber loss is not considered in (\ref{eq:ChannelRotation}), as it is compensated by a receiver amplifier. The amplifier contaminates $\bm{k}$ with amplified spontaneous emission noise, $\bm{n}\in\mathbb{C}^2$, which is a zero-mean additive white Gaussian noise with the covariance matrix $\sigma^2\bm{I}_{2\times2}$~\cite{agrawal2012fiber}. Its output is $\bm{r}=\bm{k}+\bm{n}$, hence
\[\bm{r}=\bm{H}\bm{e}+\bm{n}.\]
The angles $\theta', \gamma', \theta'',$ and $\gamma''$ are defined in the same way as their ``$\bm{e}$-domain counterparts'' as
\begin{align*}
\begin{array}{ll}
\theta'\triangleq\arg(k_xk_y^\ast), & \hspace{-0.6ex}\gamma'\triangleq\arg(k_x\cdot\D k_y^\ast),\\
\theta''\triangleq\arg(r_xr_y^\ast),& \hspace{-0.6ex}\gamma''\triangleq\arg(r_x\cdot\D r_y^\ast),
\end{array}
\end{align*}
as shown in Figs.~\ref{fig:phaseRelationK} and \ref{fig:phaseRelationR}. The relation among $\bm{e}, \bm{k}$ and $\bm{r}$ is shown in Fig.~\ref{fig:TxRxModel}. It is assumed that the receiver knows the channel parameters, {\it i.e.}, $a$ and $b$; such knowledge can be attained by transmitting training symbols at the beginning of data blocks~\cite{6936845}.
\begin{figure}
\centering
\includegraphics[scale=1]{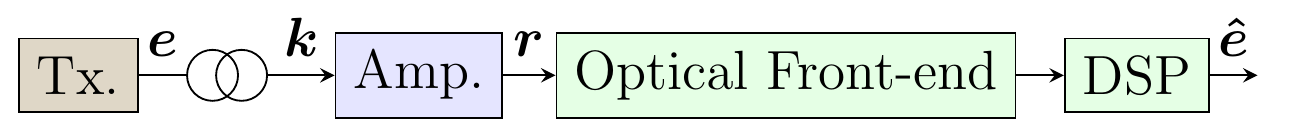}
\caption{General transceiver model.}
\label{fig:TxRxModel}
\end{figure} 

\end{subsection}

\begin{subsection}{The Optical Front-end}
\label{subsec:OpticalRx}
We use the same optical front-end as proposed in~\cite{osman} and shown in Fig.~\ref{fig:rxa}. Its components are:
\begin{itemize}
\item polarization beam splitter (PBS), which splits the input Jones vector into its $X$ and $Y$ polarizations;
\item photo-diode (PD), which transforms its input, $u\in\mathbb{C}$, to its output, $|u|^2$;
\item balanced photo-detector (BPD), which transform its two inputs, $u$ and $v\in\mathbb{C}$, to its output, $|u|^2-|v|^2$;
\item $90^\circ$ optical hybrid, which transforms its two inputs, $u$ and $v\in\mathbb{C}$, to its four outputs, $(u+v,u-v,u+iv,u-iv)$.
\end{itemize}
The outputs of the optical front-end, $w_{1:6}$, are six real-valued numbers which are processed in the back-end DSP to detect $\bm{e}$, denoted by $\bm{\hat{e}}$ (see Fig.~\ref{fig:TxRxModel}). The relation between $w_{1:6}$ and $\bm{r}$ is~\cite{tasbihi} 
\begin{align}
\begin{array}{lll}
\hspace{-0.6ex}w_1=|r_x|^2,& w_3= 2\Re(r_xr_y^\ast),& w_5=2\Re(r_x\cdot\D r_y^\ast),\\
\hspace{-0.6ex}w_2= |r_y|^2,& w_4= 2\Im(r_xr_y^\ast),& w_6= 2\Im(r_x\cdot\D r_y^\ast).\\ 
\end{array}
\label{eq:w16relation}
\end{align} 
\begin{figure}
\centering
\includegraphics[scale=1]{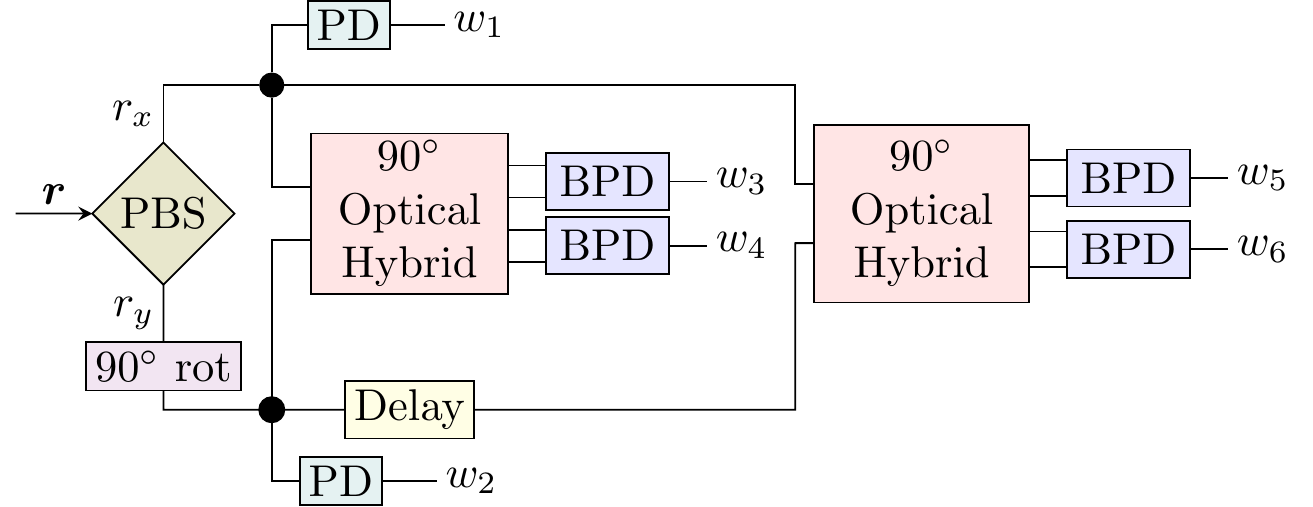}
\caption{The optical front-end, presented in~\cite{osman}.}
\label{fig:rxa}
\end{figure}

\end{subsection}
\end{section}

\begin{section}{ML Detection}
\label{sec:ML}
In this section, we derive the ML detector for the transmitted data under three processing assumptions, after observing $w_{1:6}$. Based on the transmitted and the received quantities, we define the vectors
\begin{equation*}
\begin{aligned}
\bm{d}_e&\triangleq[|e_x|,|e_y|e^{i\theta},|\D e_y|e^{i\gamma}]^t\in\mathbb{R}^+\times\mathbb{C}^2,\\
\bm{d}_k&\triangleq[|k_x|,|k_y|e^{i\theta'},|\D k_y|e^{i\gamma'}]^t\in\mathbb{R}^+\times\mathbb{C}^2,\\
\bm{d}_r&\triangleq[|r_x|,|r_y|e^{i\theta''},|\D r_y|e^{i\gamma''}]^t\in\mathbb{R}^+\times\mathbb{C}^2,\\
\hat{\bm{d}}_k&\triangleq[|k_x|,|k_y|e^{i\theta'}]^t\in\mathbb{R}^+\times\mathbb{C},
\end{aligned}
\end{equation*}
and
\begin{equation*}
\hat{\bm{d}}_r\triangleq[|r_x|,|r_y|e^{i\theta''}]^t\in\mathbb{R}^+\times\mathbb{C}.
\end{equation*}
Note that $w_{1:6}$ and $\bm{d}_r$ are in one-to-one correspondence. The relationship among $\langle\bm{d}_k,\bm{d}_r\rangle, \langle\hat{\bm{d}}_k,\hat{\bm{d}}_r\rangle$, and the components of $\bm{d}_k$ and $\bm{d}_r$ is depicted in Fig.~\ref{fig:triangles}, to be used in later sections. 
\begin{figure}
\centering
\includegraphics[scale=1]{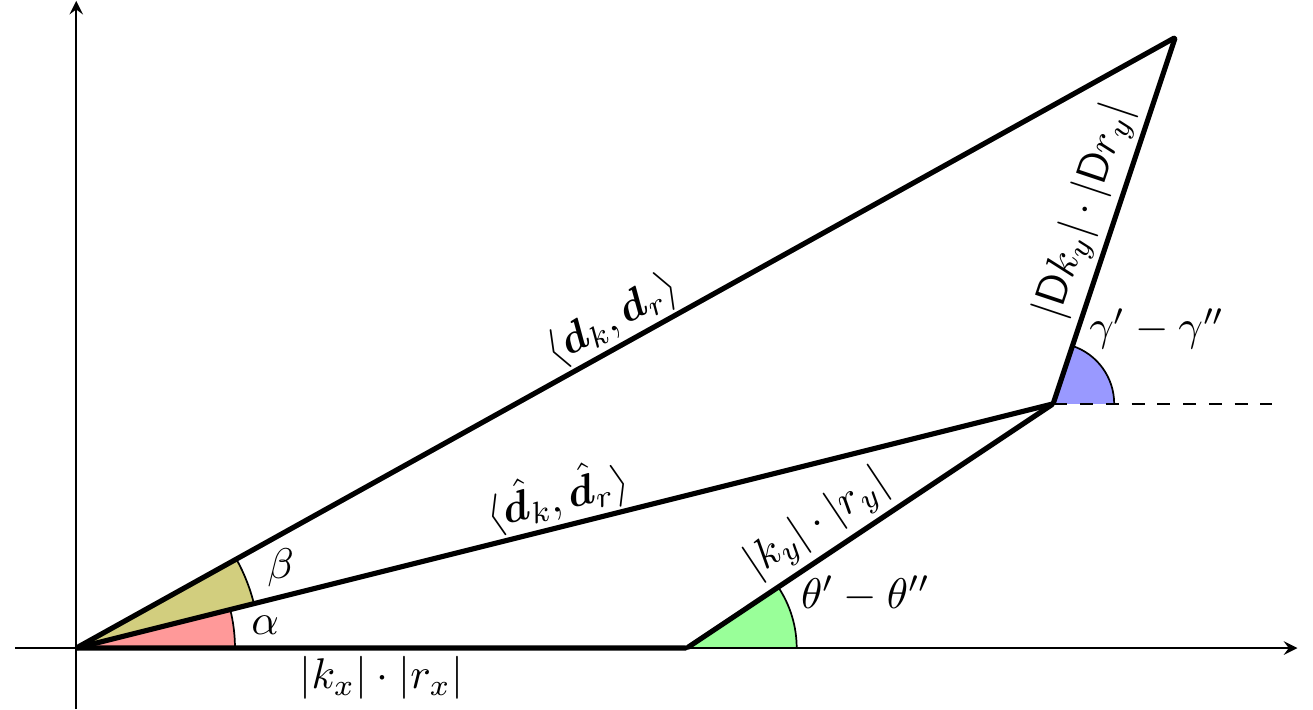}
\caption{Relationship among $\langle\bm{d}_k,\bm{d}_r\rangle, \langle\hat{\bm{d}}_k,\hat{\bm{d}}_r\rangle$, and the components of $\bm{d}_k$ and $\bm{d}_r$ in the complex plane.}
\label{fig:triangles}
\end{figure}

\begin{subsection}{Symbol-by-symbol ML Detection}
\label{sec:JointDetection}
In this section, we discuss about the symbol-by-symbol detection of all four dimensions. As $\mathbb{E}$\ding{214}$\mathbb{K}$, for the ease of computation, first we decide on $\bm{d}_k$; after that, by a bijection we find $\bm{d}_e$. In this process $|\D k_y|$ is fixed as it is decoded in the previous time slot. 

By the definition of conditional PDF, we can write the likelihood function as 
\begin{align}
& f(|r_x|,|r_y|,\theta'',\gamma''\mid \bm{d}_k,|\D r_y|)=\nonumber\\
& f(|r_x|,|r_y|\mid \bm{d}_k,|\D r_y|)\cdot f(\theta'',\gamma''\mid \bm{d}_k,|\D r_y|,|r_x|,|r_y|).
\label{eq:conditional}
\end{align}
In Theorems~\ref{thrm:radii} and~\ref{thrm:phase}, we find $f(|r_x|,|r_y|\mid \bm{d}_k,|\D r_y|)$ and $f(\theta'',\gamma''\mid \bm{d}_k,|\D r_y|,|r_x|,|r_y|)$, respectively. After that, we can find the likelihood function by~(\ref{eq:conditional}).
\begin{theorem}
\label{thrm:radii}
\begin{equation}
f(|r_x|,|r_y|\mid \bm{d}_k,|\D r_y|)=
\frac{|r_x|\cdot|r_y|}{\sigma^4}\exp\left(\frac{-(|r_x|^2+|r_y|^2+|k_x|^2+|k_y|^2)}{2\sigma^2}\right)
\cdot\mathcal{I}_0(\lambda_x)\mathcal{I}_0(\lambda_y),
\label{eq:thrm1}
\end{equation}
where $\mathcal{I}_0(\cdot)$ denotes the modified Bessel function of order zero and $\lambda_u=\frac{|r_u|\cdot|k_u|}{\sigma^2}$ for $u\in\{x,y\}$.
\end{theorem}
\begin{proof}
Let $R_u=|R_u|e^{i\Psi_u}$ and $K_u=|K_u|e^{i\Phi_u}, u\in\{x,y\},$ be the random variables representing the received and the transmitted (in $\mathbb{K}$ domain) signals. Note that $R_x$ and $R_y$ are independent complex Gaussian random variables, with means $k_u=|k_u|e^{i\phi_u}$ and each one has a covariance matrix $\sigma^2\bm{I}_{2\times2}$. The radius has a Rician distribution with the PDF given as~\cite{simon}
\begin{equation}
f(|r_u|~\mid~ |k_u|)=\frac{|r_u|}{\sigma^2}\exp\left(\frac{-(|r_u|^2+|k_u|^2)}{2\sigma^2}\right)\mathcal{I}_0(\lambda_u).
\label{eq:Rician}
\end{equation}
In addition, given $|K_u|$, $|R_u|$ is independent of other parameters in $(\bm{d}_k, |\D r_y|)$, so we have
\[f(|r_x|,|r_y|\mid \bm{d}_k,|\D r_y|)=f(|r_x|\mid |k_x|)\cdot f(|r_y|\mid |k_y|),\]
from which~(\ref{eq:thrm1}) is obtained by substituting from~(\ref{eq:Rician}).
\end{proof}
\begin{theorem}
\label{thrm:phase}
\begin{equation*}
f(\theta'',\gamma''\mid \bm{d}_k,|\D r_y|,|r_x|,|r_y|)=\frac{\mathcal{I}_0(\frac{|\langle\bm{d}_k,\bm{d}_r\rangle|}{\sigma^2})}{4\pi^2\mathcal{I}_0(\lambda_x)\mathcal{I}_0(\lambda_y)\mathcal{I}_0(\D\lambda_y)}.
\end{equation*}
\end{theorem}
\begin{proof}
See Appendix~\ref{app0}.
\end{proof}
By using Theorems~\ref{thrm:radii} and~\ref{thrm:phase}, and (\ref{eq:conditional}), we have
\begin{equation} 
 f(|r_x|,|r_y|,\theta'',\gamma''\mid \bm{d}_k,|\D r_y|)=
\frac{|r_x|\cdot|r_y|}{4\pi^2\sigma^4}\frac{\mathcal{I}_0\left(\frac{|\langle\bm{d}_k,\bm{d}_r\rangle|}{\sigma^2}\right)\exp\left(\frac{-(|r_x|^2+|r_y|^2+|k_x|^2+|k_y|^2)}{2\sigma^2}\right)}{\mathcal{I}_0(\D\lambda_y)}.
\label{eq:4d}
\end{equation}
Note that $|\langle\bm{d}_k,\bm{d}_r\rangle|$ is the magnitude of the correlation of the observation, $\bm{d}_r$, and the hypothesis, $\bm{d}_k$ (see Fig.~\ref{fig:triangles}.) 

To do ML symbol-by-symbol detection, we must solve 
\begin{equation}
\begin{array}{cc}
\underset{\bm{d}_k}{\arg\max}& f(w_{1:6} \mid \bm{d}_k,|\D r_y|)\\
\mbox{subject to}&|\bm{d}_k(3)|=|\D k_y|.
\end{array}
\label{eq:actualopt}
\end{equation} 
From (\ref{eq:w16relation}), we see that given $|\D r_y|$, there is a bijection between $w_{1:6}$ and $(|r_x|,|r_y|,\theta'',\gamma'')$, which allows us to rewrite (\ref{eq:actualopt}) as
\begin{equation}
\begin{array}{cc}
\underset{\bm{d}_k}{\arg\max} & f(|r_x|,|r_y|,\theta'',\gamma''\mid \bm{d}_k,|\D r_y|)\\
\mbox{subject to} & |\bm{d}_k(3)|=|\D k_y|.
\end{array}
\label{eq:actualopt2}
\end{equation}
Noting that $|\bm{d}_k(3)|$ is constant, by (\ref{eq:4d}) and eliminating the common factors among all hypotheses, (\ref{eq:actualopt2}) is equivalent to 
\begin{equation}
\begin{array}{cc}
\underset{\bm{d}_k}{\arg\min} & |\bm{d}_k|^2-2\sigma^2\ln\left(\mathcal{I}_0\left(\frac{|\langle\bm{d}_k,\bm{d}_r\rangle|}{\sigma^2}\right)\right)\\
\mbox{subject to} & |\bm{d}_k(3)|=|\D k_y|,
\end{array}
\label{eq:actualopt3}
\end{equation}
which in practice can be solved by examining all possible $\bm{d}_k$'s that satisfy the condition.

After finding $\bm{d}_k$ from (\ref{eq:actualopt3}), we find the equivalent $\bm{d}_e$. Note that~\cite{osman}
\begin{equation}
\label{eq:noisefree3D}
\left[\begin{array}{c}
|k_x|^2\\|k_y|^2\\2|k_x|\cdot|k_y|\cos(\theta')\\2|k_x|\cdot|k_y|\sin(\theta')
\end{array}\right]=\bm{M}\left[\begin{array}{c}
|e_x|^2\\
|e_y|^2\\
2|e_x|\cdot|e_y|\cos(\theta)\\
2|e_x|\cdot|e_y|\sin(\theta)
\end{array}\right],
\end{equation}
where $\bm{M}=$
\[\left[\begin{array}{cccc}
|a|^2 & |b|^2 & \Re(ab^\ast) & -\Im(ab^\ast)\\
|b|^2 & |a|^2 & -\Re(ab^\ast) & \Im(ab^\ast)\\
-2\Re(ab) & 2\Re(ab) & \Re(a^2-b^2) & -\Im(a^2+b^2)\\
-2\Im(ab) & 2\Im(ab) & \Im(a^2-b^2) & \Re(a^2+b^2)
\end{array}\right].\]
Hence, after decoding $(|k_x|,|k_y|,\theta')$, by using (\ref{eq:noisefree3D}) we can find $(|e_x|,|e_y|,\theta)$. To decide on $\gamma$, we note that~\cite{tasbihi}
\begin{equation}
\label{eq:noisefree4D3}
|k_x|\cdot|\D k_y|\exp(i\gamma')=\exp(i\gamma)\left[\begin{array}{cccc}
a^2 & -b^2 & -ab & ab
\end{array}\right]\bm{\ell},
\end{equation}
where
\[\bm{\ell}=\left[\begin{array}{c}
|e_x|\cdot|\D e_y|\\
|e_y|\cdot|\D e_x|\cdot\exp(-i(\theta+\D\theta))\\
|e_x|\cdot|\D e_x|\cdot\exp(-i\D\theta)\\
|e_y|\cdot|\D e_y|\cdot\exp(-i\theta)
\end{array}\right].\]
After decoding $\gamma'$ from (\ref{eq:actualopt3}) and finding $(|e_x|,|e_y|,|\theta|)$ from (\ref{eq:noisefree3D}), the only unknown in (\ref{eq:noisefree4D3}) is $\gamma$, which can easily be solved for.
\paragraph*{High SNR approximation} For large arguments, $\mathcal{I}_0(\cdot)$ can be approximated as~\cite[eq.~9.7.1]{abramowitz1965handbook} 
\begin{equation}
\mathcal{I}_0(x)\simeq\frac{e^x}{\sqrt{2\pi x}}\left(1+\mathcal{O}(x^{-1})\right).
\label{eq:i0}
\end{equation}
By neglecting $\mathcal{O}(x^{-1})$ terms in~(\ref{eq:i0}) and noting that
\begin{equation}\lim_{x\rightarrow\infty}\frac{x-\ln\left(\sqrt{2\pi x}\right)}{x}=1,\label{eq:asymlim}\end{equation}
we can approximate~(\ref{eq:actualopt3}) at high SNRs as
\begin{equation}
\begin{array}{cc}
\underset{\bm{d}_k}{\arg\min} & |\bm{d}_k|^2-2|\langle\bm{d}_k,\bm{d}_r\rangle|\\
\mbox{subject to} & |\bm{d}_k(3)|=|\D k_y|.
\end{array}
\label{eq:actualopt4}
\end{equation} 
Despite the ``similarity'' of~(\ref{eq:actualopt4}) and the minimum-distance decoder, they behave differently. By eliminating $|\bm{d}_r|^2$ in the expansion of $|\bm{d}_k-\bm{d}_r|^2$, the minimum-distance decoder solves 
\begin{equation*}
\begin{array}{cc}
\underset{\bm{d}_k}{\arg\min} & |\bm{d}_k|^2-2\Re\left(\langle\bm{d}_k,\bm{d}_r\rangle\right)\\
\mbox{subject to} & |\bm{d}_k(3)|=|\D k_y|,
\end{array}
\end{equation*}  
which is different from~(\ref{eq:actualopt4}). For example, for
$\bm{d}_{k1}=\left[3.5,1.2e^{-3i},2.8e^{-2i}\right]^t$,
$\bm{d}_{k2}=\left[2.1,2.8e^{0.5i},2.8e^{3i}\right]^t$,
$\bm{d}_{r}=\left[0.6,1.8e^{-2i},1.8e^{3i}\right]^t$,
we have
\[|\bm{d}_{k1}|^2-2|\langle\bm{d}_{k1},\bm{d}_r\rangle|<|\bm{d}_{k2}|^2-2|\langle\bm{d}_{k2},\bm{d}_r\rangle|,\]
while
\[|\bm{d}_{k1}|^2-2\Re\left(\langle\bm{d}_{k1},\bm{d}_r\rangle\right)>|\bm{d}_{k2}|^2-2\Re\left(\langle\bm{d}_{k2},\bm{d}_r\rangle\right).\]
As a result, the minimum-distance decoder maps $\bm{d}_r$ to $\bm{d}_{k2}$, while the ML decoder maps it to $\bm{d}_{k1}$.

\end{subsection}
\begin{subsection}{Sequence ML Detection}
\label{sec:MaxProd}
In this section, we show how to decode a sequence of transmitted data by using the min-sum algorithm on the factor-graph of the system~\cite{factorgraph}. The sequence comprises two types of symbols: a pilot symbol which is sent at the beginning of the sequence, and is known to the receiver (see Sec.~\ref{subsec:Transmitter}), and data symbols.\\

The flow-graph of the system for a sequence of length four is shown at Fig.~\ref{fig:flowgraph}. {\it Variable nodes} (v-nodes) and {\it check nodes} (c-nodes) are shown with circles and rectangles respectively. The channel matrix is represented by $\bm{H}$ and the optical front-end is denoted by $h_i, i\in\{0,1,2\}$.\\ 
\begin{figure}
\centering
\includegraphics[scale=1]{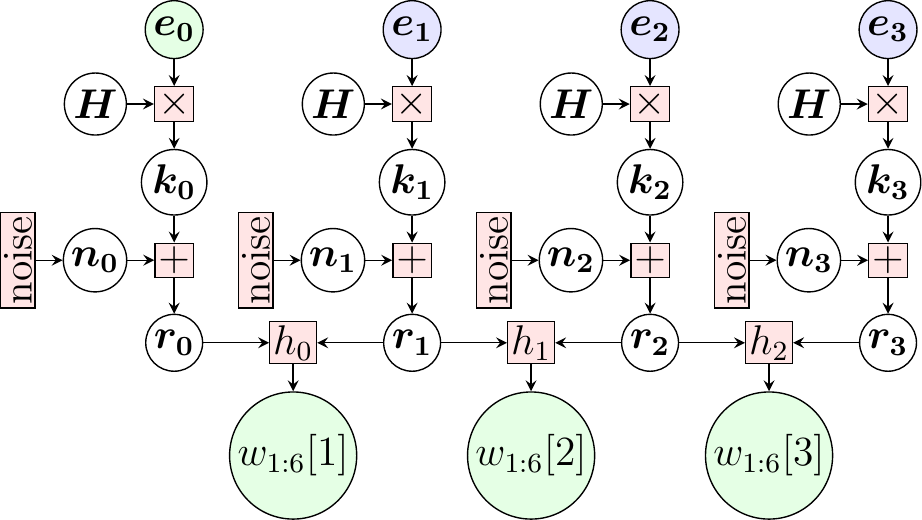}
\caption{Flow-graph of the system. $\bm{H}$ is the channel matrix and $h_i$ denotes the optical front-end operation.}
\label{fig:flowgraph}
\end{figure}

Solving {\it ML sequence-detection} (MLSD) problem requires finding
\[\underset{\bm{d}_e[1,\ldots,n]}{\arg\max}~~f(w_{1:6}[1,\ldots,n]\mid\bm{d}_e[0,\ldots,n]),\]
which, as $\mathbb{E}$\ding{214}$\mathbb{K}$, can be written as
\begin{equation}
\underset{\bm{d}_k[1,\ldots,n]}{\arg\max}~~f(w_{1:6}[1,\ldots,n]\mid\bm{d}_k[0,\ldots,n]),
\label{eq:JointOpt}
\end{equation}
where $\bm{d}_e[0]$ and $\bm{d}_k[0]$ are the pilot symbols. Note that $w_{1:6}[1,\ldots,n]$ and $\bm{d}_k[0,\ldots,n]$ form a second-order hidden Markov chain and as a result 
\[f(w_{1:6}[1,\ldots,n]\mid\bm{d}_k[1,\ldots,n])=\prod_{j=1}^{n}f(w_{1:6}[j]\mid\bm{d}_k[j-1,j]).\]
By using~(\ref{eq:4d}), (\ref{eq:JointOpt}) is equivalent to  
\begin{equation}
\underset{\bm{d}_k[1,\ldots,n]}{\arg\min}~\sum_{j=1}^{n}|k_x[j]|^2+|k_y[j]|^2-2\sigma^2\ln\left(\frac{\mathcal{I}_0\left(\frac{|\langle\bm{d}_k[j],\bm{d}_r[j]\rangle|}{\sigma^2}\right)}{\mathcal{I}_0(\lambda_y[j-1])}\right),
\label{eq:MinSumObjective}
\end{equation}
which suggests to use the min-sum algorithm on its factor graph. A factor-graph representation of the objective function of (\ref{eq:MinSumObjective}) for a sequence of length four is shown in Fig.~\ref{fig:Final_FG}. Note that the factor graph is cycle-free, hence the min-sum algorithm produces the exact minimum. In addition, it is equivalent to the Viterbi algorithm~\cite{factorgraph}. 
\paragraph*{High SNR approximation}From~(\ref{eq:i0}) and~(\ref{eq:asymlim}), (\ref{eq:MinSumObjective}) can be approximated at high SNRs as
\begin{equation*}
\underset{\bm{d}_k[1,\ldots,n]}{\arg\min}~\sum_{j=1}^{n}|k_x[j]|^2+|k_y[j]|^2-2|\langle\bm{d}_k[j],\bm{d}_r[j]\rangle|+2|k_y[j-1]|\cdot|r_y[j-1]|.
\end{equation*}
\begin{figure}
\centering
\includegraphics[scale=1]{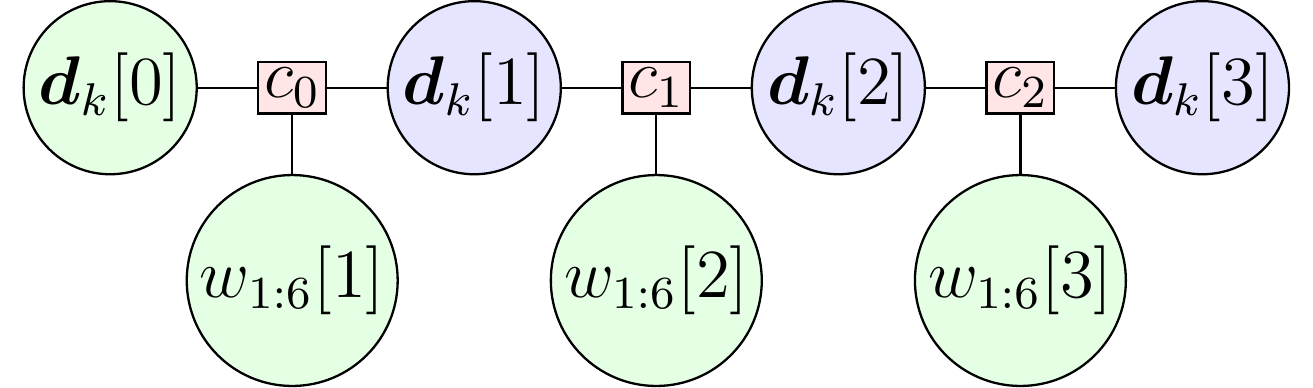}
\caption{Final factor-graph. The c-nodes, $c_i$, are cost functions.}
\label{fig:Final_FG}
\end{figure} 

\end{subsection}
\begin{subsection}{Successive ML Detection}
\label{sec:SuccessiveDetection}
In symbol-by-symbol detection, we search over all possible transmitted symbols for the detection of each received symbol; {\it e.g.}, for an $n_r$-ring/$n_p$-ary phase constellation for each polarization, $(n_rn_p)^2$ different possibilities must be examined. We can reduce this complexity by decoding in a successive manner, at the expense of an increase in SER. In the proposed successive detection method, we decode the first three dimensions jointly, then proceed to decode the fourth one with the knowledge of the first three dimensions. In this way, we must examine $(n_r^2+1)n_p$ different possibilities; which, for a large value of $n_p$, the complexity reduction is significant. For example, for an $8$-ring/$16$-ary phase constellation, we must search over $2^{14}$ possibilities to decode each symbol in the symbol-by-symbol scheme, while this number reduces to $2^{10}+16$ for successive detection.
\begin{subsubsection}{The Likelihood Function at the First Successive-Step}
The ML detection of the first three dimensions necessitates solving 
\begin{equation*}
\underset{\hat{\bm{d}}_k}{\arg\max}~~f(w_{1:4}\mid\hat{\bm{d}}_k),
\end{equation*}
which, due to the bijection between $\hat{\bm{d}}_r$ and $w_{1:4}$ (see (\ref{eq:w16relation})), can be written as
\begin{equation}
\underset{\hat{\bm{d}}_k}{\arg\max}~f(w_{1:4} \mid \hat{\bm{d}}_k)=\underset{\hat{\bm{d}}_k}{\arg\max}~f(|r_x|,|r_y|,\theta''\mid \hat{\bm{d}}_k).
\label{eq:suc}
\end{equation}
In Theorem~\ref{thrm:suc1}, we find the likelihood function of the first three dimensions. 
\begin{theorem}
\label{thrm:suc1}
\begin{equation*}
f(|r_x|,|r_y|,\theta''\mid \hat{\bm{d}}_k)=
\frac{|r_x|\cdot|r_y|}{2\pi\sigma^4}\exp\left(\frac{-\left(|\hat{\bm{d}}_r|^2+|\hat{\bm{d}}_k|^2\right)}{2\sigma^2}\right)\mathcal{I}_0\left(\frac{|\langle\hat{\bm{d}}_k,\hat{\bm{d}}_r\rangle|}{\sigma^2}\right).\end{equation*}
\end{theorem}
\begin{proof}
See Appendix~\ref{app1}.
\end{proof}
By using Theorem~\ref{thrm:suc1}, we can rewrite (\ref{eq:suc}) as 
\begin{equation}
\underset{\hat{\bm{d}}_k}{\arg\min}~|\hat{\bm{d}}_k|^2-2\sigma^2\ln\left(\mathcal{I}_0\left(\frac{|\langle\hat{\bm{d}}_k,\hat{\bm{d}}_r\rangle|}{\sigma^2}\right)\right).
\label{eq:suc1cost}
\end{equation}
Note that the optimal $\theta'$, obtained by solving (\ref{eq:suc1cost}), is the closest possible one to $\theta''$, {\it i.e.,} it maximizes $\cos(\theta'-\theta'')$ over all feasible $\theta'$. That is because $\mathcal{I}_0(\cdot)$ and $\ln(\cdot)$ are strictly increasing functions. Hence, to minimize the objective function of~(\ref{eq:suc1cost}), we must maximize $|\langle\hat{\bm{d}}_k,\hat{\bm{d}}_r\rangle|$, which from Fig.~\ref{fig:triangles} obtains when $\theta'-\theta''$ is the ``closest'' one to zero, {\it i.e.}, $\cos(\theta'-\theta'')$ must be maximized. 
\paragraph*{High SNR approximation} Similar to the symbol-by-symbol detection, at high SNRs, (\ref{eq:suc1cost}) can be approximated as 
\begin{equation*}
\underset{\hat{\bm{d}}_k}{\arg\min}~|\hat{\bm{d}}_k|^2-2|\langle\hat{\bm{d}}_k,\hat{\bm{d}}_r\rangle|.
\end{equation*}
\end{subsubsection}
\begin{subsubsection}{The Second Successive-Step}
\label{subsec:secondSuccessive}
At the second successive-step, the fourth dimension is decoded. According to (\ref{eq:w16relation}) we have
\begin{align*}
w_5&=2\Re\left(r_x\cdot\D r_y^\ast\right)=2|r_x|\cdot|\D r_y|\cdot\cos(\gamma''),\\
w_6&=2\Im\left(r_x\cdot\D r_y^\ast\right)=2|r_x|\cdot|\D r_y|\cdot\sin(\gamma'').
\end{align*}
At this step, the intensities are treated as constants, as they have been decoded at the first successive-step. As a result, there is a bijection between $w_{5:6}$ and $\gamma''$. The decoder performs ML detection of $\gamma'$ by solving
\begin{align}
&\underset{\gamma'}{\arg\max}~f(w_{5:6}\mid\bm{d}_k,\hat{\bm{d}}_r,|\D r_y|)=\nonumber\\
&\underset{\gamma'}{\arg\max}~f(\gamma''\mid\bm{d}_k,\hat{\bm{d}}_r,|\D r_y|).
\label{eq:suc4}
\end{align}
Theorem~\ref{thrm:fourth} provides an easy way to decide on $\gamma'$.
\begin{theorem}
\label{thrm:fourth}
\begin{equation}
\underset{\gamma'}{\arg\max}~f(\gamma''\mid\bm{d}_k,\hat{\bm{d}}_r,|\D r_y|) = \underset{\gamma'}{\arg\max}~\cos(\gamma'-\gamma''-\alpha),
\label{eq:final4D}
\end{equation} 
where $\alpha=\arg\left(\langle\hat{\bm{d}}_k,\hat{\bm{d}}_r\rangle\right)$, as shown in Fig.~\ref{fig:triangles}.
\end{theorem}
\begin{proof}
See Appendix~\ref{app2}.
\end{proof}
The interpretation of (\ref{eq:final4D}) is that the decoder chooses the closest feasible $\gamma'$ to $\gamma''+\alpha$. This can be justified by using Fig.~\ref{fig:triangles} and (\ref{eq:actualopt3}) as well. In the first successive step we have decoded $\hat{\bm{d}}_k$ and as a result, $(\lambda_x,\lambda_y,\theta'-\theta'')$ are fixed. From (\ref{eq:actualopt3}), we must maximize $|\langle\bm{d}_k,\bm{d}_r\rangle|$ to minimize its objective function, which happens when the segment $\langle\bm{d}_k,\bm{d}_r\rangle$ in Fig.~\ref{fig:triangles} has the smallest angular deviation from the segment $\langle\hat{\bm{d}}_k,\hat{\bm{d}}_r\rangle$. This means that $\gamma'-\gamma''$ must be the closest one to $\alpha$.    
\end{subsubsection}

\end{subsection}
\end{section}
\begin{section}{Subchannel Fading}
\label{sec:fading}
In this section, we show that the optical front-end, studied in this paper, causes the fourth DOF ($\gamma=\arg(e_x\cdot\D e_y^\ast)$) to be subjected to fast fading, which makes this subchannel exhibit a symbol error rate behaviour that is markedly different than the other subchannels. For the purpose of this discussion, we neglect the effect of noise (setting noise to zero); instead, we focus on the effect that the channel matrix, $\bm{H}$, has on the four DOF subchannels.

As we see from (\ref{eq:noisefree3D}), the relationship between the input and the output of the first three DOF subchannels is determined by $\bm{H}$, and does not depend on previously transmitted symbols. As a result, the subchannels of the first three dimensions change only block-by-block; hence, we expect the first three dimensions to experience a slow (block) fading channel. 

The output of the fourth subchannel, however, is not only a function of the channel parameters, but also a function of the previously transmitted symbols as well. Particularly, from~(\ref{eq:noisefree4D3}), we have
\begin{align*}
e^{i\gamma'}=\frac{c}{|k_x|\cdot|\D k_y|}e^{i\gamma},
\end{align*}
where, the complex number $c$ is
\begin{align*}
\begin{multlined} 
c= a^2|e_x|\cdot|\D e_y|-b^2|e_y|\cdot|\D e_x|e^{-i\left(\theta+\D\theta\right)}\\
-ab|e_x|\cdot|\D e_x|e^{-i\D\theta}+ab|e_y|\cdot|\D e_y|e^{-i\theta}.
\end{multlined}
\end{align*}
The coefficient $c$ is a function of $(|\D e_x|,|\D e_y|,\D\theta)$, which makes the fourth DOF subchannel vary symbol-by-symbol. As a result, we see that the fourth DOF subchannel suffers from fast fading; and similar to a Rayleigh fading channel, we expect the symbol error rate of the fourth dimension to be proportional to $\frac{1}{\text{SNR}}$~\cite[p. 533]{wozencraft}. In Sec.~\ref{sec:Results}, we see that the SER figures support this claim.  

\end{section}
\begin{section}{Numerical Results}
\label{sec:Results}
In this section, we provide some numerical results to compare the discussed detection methods. For all figures, the SNR is defined as the average transmitted energy per polarization over the complex-noise variance per polarization. Specifically, for the discussed constellation, the SNR is defined as
\[\text{SNR}=\frac{r_1^2(1+\delta^2\left(\frac{n_r-1}{2}\right))}{2\sigma^2}.\]

The resulting SERs for different constellations are shown in Figs.~\ref{fig:SER24_114}--\ref{fig:SER816_545}. The channel matrix for each block of data is chosen uniformly over all possible $\bm{H}$ matrices (see Sec.~\ref{subsec:Channel}). By increasing $\delta^2$, the rings become more distant from each other, hence it improves the performance of the first two dimensions in terms of SER, but there is a trade-off with the performance of the phase channels. For example, for $2$-ring/$4$-ary phase constellation and the target SER of $10^{-3}$, changing $\delta^2$ from $1$ to $4.83$ results an improve of $7$~dB in the intensity channels, while it degrades the performance of the third channel by $3.5$~dB (compare Figs.~\ref{fig:SER24_114} and~\ref{fig:SER24_214}). As another example, for $8$-ring/$8$-ary phase constellation and the same target SER, by changing $\delta^2$ from $2.12$ to $15.36$, the performance of the first two dimensions improve by $1.5$~dB, while the performance of the third and the fourth dimensions degrades by $8$~dB and $2.5$~dB, respectively (compare Figs.~\ref{fig:SER88_214}--\ref{fig:SER88_1042}).

As shown, while the complexity of successive detection is smaller than other methods, its SER does not differ noticeably. For example, by using $8$-ring/$8$-ary phase constellation with successive detection, the complexity of brute-force search is reduced (approximately) by a factor of 8, while for the target SER of $10^{-3}$, it degrades the performance of the third channel less than $0.5$~dB for $\delta^2=0.69$, and does not affect the performance of other dimensions noticeably (see Fig.~\ref{fig:SER88_214}).

From Figs.~\ref{fig:SER24_114}--\ref{fig:SER88_1042}, we see that the high-SNR approximation is a very good approximation which, by avoiding computing the modified Bessel function, can reduce the complexity of decoder. In all of these figures, the approximated figures and the actual ones are almost superposed. As a result, due to the large size of the constellation and to remedy the long running time of the simulation, the SER figures for $8$-ring/$16$-ary phase constellation are computed by high-SNR approximated formula.

As discussed in Sec.~\ref{sec:fading}, we expect the fourth dimension to be under fast fading. The results support our expectation. Similar to a Rayleigh fading channel, the SER of the fourth dimension is proportional to $\frac{1}{\text{SNR}}$, and that is the reason of the linear behaviour of the fourth dimension symbol error rate, shown in Figs.~\ref{fig:SER24_114}--\ref{fig:SER816_545}. The fast fading behaviour is due to the non-zero $b$ entry in the $\bm{H}$ matrix, which entangles the fourth channel with the past data. Hence, we expect the fourth dimension to behave as same as the third dimension when $b=0$. Fig.~\ref{fig:SER_norotate} shows that this is indeed true. For this figure, the channel matrix varies block-by-block, but in all cases, its $b$ entry is zero. As $|a|^2+|b|^2=1$, the no-entanglement condition implies that $a = e^{i\zeta}$, for some random $\zeta\in[-\pi,\pi)$. 

The achievable rate for different constellations and $\delta^2$ are shown in Figs.~\ref{fig:rate24}--\ref{fig:rate88}. The rates are actually
\begin{equation}I(|K_x|,|K_y|,\Theta',\Gamma';|R_x|,|R_y|,\Theta'',\Gamma''\mid |\D K_y|,|\D R_y|),\label{eq:rate}\end{equation}
where $I(U;V)$ denotes the mutual information between the random variables $U$ and $V$,
and is computed by the Monte Carlo method. As there is a conditioning on $|\D R_y|$,~(\ref{eq:rate}) is actually the achievable rate of the scheme, where the receiver feeds back the intensity of the received Y polarization. As we are using an $n_r$-ring/$n_p$-ary phase constellation, the maximum rate is $2\log(n_rn_p)$ bits per channel use.

The $\delta^2$ which causes the minimum Euclidean distance between two points on a ring (which happens for the inner-most ring) to be the same as the minimum Euclidean distance between two points on different rings (which happens for the two outer-most rings) is denoted by $\delta^2_{\text{bl}}$ (balanced $\delta^2$). It can be easily shown that 
\begin{equation}
\begin{multlined}
\delta^2_{\text{bl}}=4\sin^2\left(\frac{\pi}{n_p}\right)(2n_r-3)\\+4\sin\left(\frac{\pi}{n_p}\right)\sqrt{4(n_r-1)(n_r-2)\sin^2\left(\frac{\pi}{n_p}\right)+1}.
\end{multlined}
\label{eq:optimal_rho}
\end{equation}
As it is shown in Figs.~\ref{fig:rate24}--\ref{fig:rate88}, although $\delta^2_{\text{bl}}$ is not the optimal $\delta^2$, it is nearly optimal. The $\delta^2_{\text{bl}}$ for different constellation are shown in Table~\ref{tab:rho}. Inspired by~\cite{ungerboeck1982channel}, in Fig.~\ref{fig:rate} we have compared the achievable rate of different constellations at their $\delta^2_{\text{bl}}$. This figure shows the necessity of using an error-correcting code at the encoder. For example, by using a 4-ring/8-ary phase constellation without any error-correcting code, we can transmit $10$ bits per channel use at the SNR of $25$~dB, while by using an 8-ring/8-ary phase constellation and a code of rate $\frac{5}{6}$, we can achieve the same rate at the SNR of $20$~dB; hence, we can save $5$~dB. 
\begin{table}
\caption{The $\delta^2_{\text{bl}}$ for different constellation, computed by (\ref{eq:optimal_rho}).}
\centering
\begin{tabular}{|c|c|c|c|c|c|}
\hline
$(n_r,n_p)$ & $\delta^2_{\text{bl}}$ & $(n_r,n_p)$ &$\delta^2_{\text{bl}}$ & $(n_r,n_p)$ & $\delta^2_{\text{bl}}$\\
\hline
$(2,4)$ & $4.83$ & $(4,4)$ & $20.20$& $(4,8)$ & $6.18$\\
\hline
$(8,4)$ & $52.08$& $(8,8)$ & $15.36$ & $(8,16)$ & $4.10$\\
\hline
\end{tabular}
\label{tab:rho}
\end{table}
\begin{figure}
\centering
\includegraphics[scale=1]{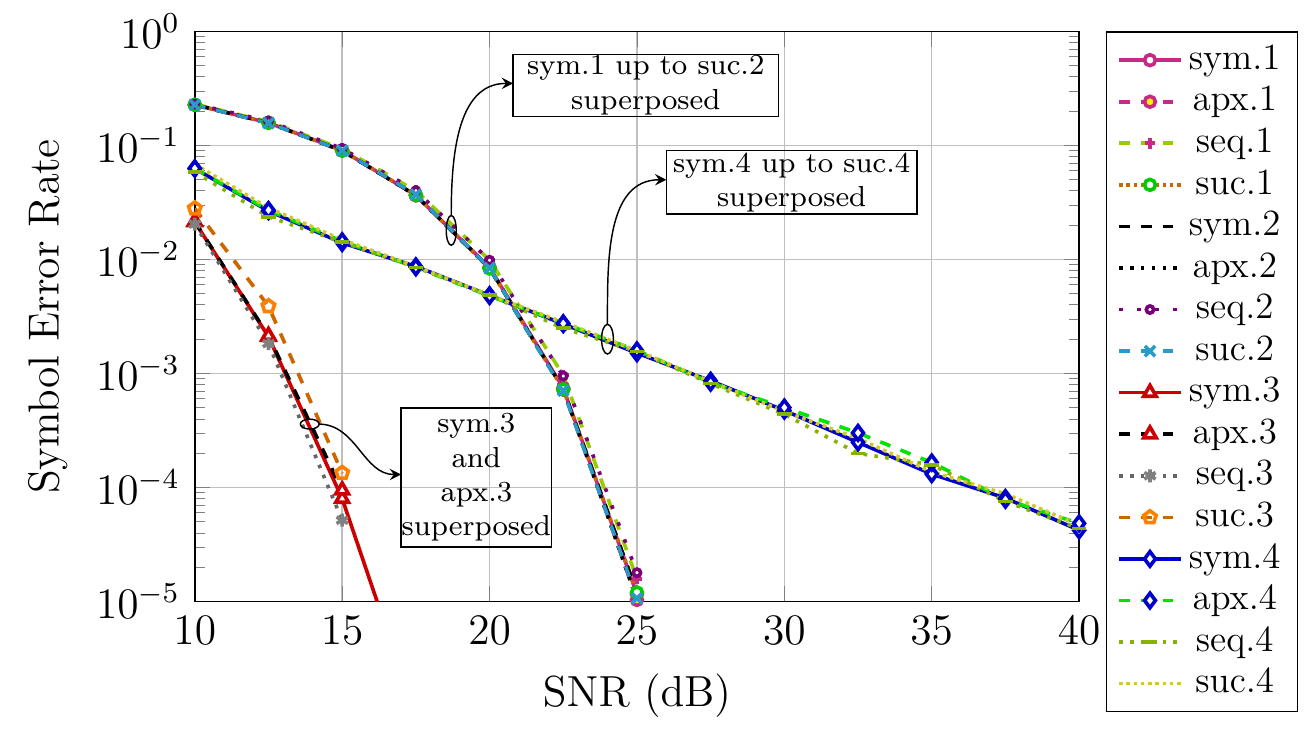}
\caption{Symbol-by-symbol detection (sym), its high-SNR approximation (apx), and successive detection (suc) of four dimensions for 2-ring/4-ary phase constellation and $\delta^2=1$.}
\label{fig:SER24_114}
\end{figure}
\begin{figure}
\centering
\includegraphics[scale=1]{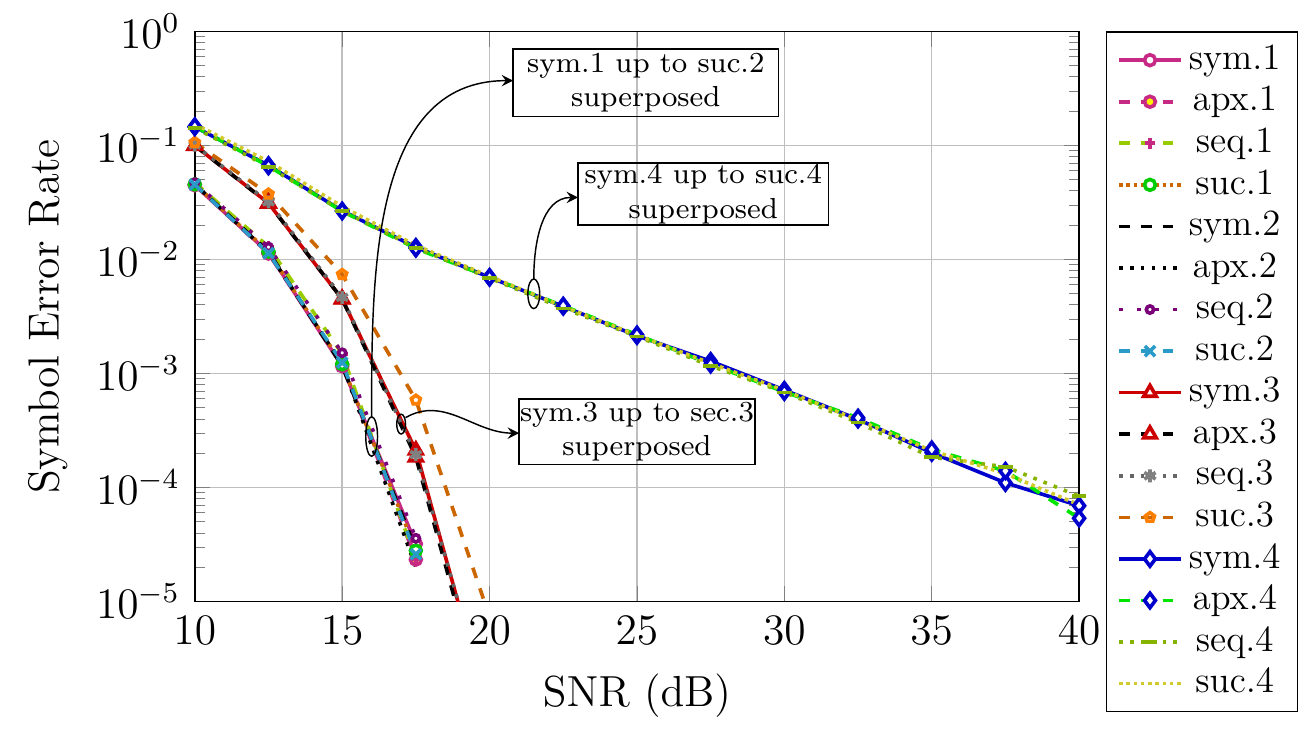}
\caption{Symbol-by-symbol detection (sym), its high-SNR approximation (apx), sequence detection (seq), and successive detection (suc) of four dimensions for 2-ring/4-ary phase constellation and $\delta^2=2(1+\sqrt{2})\simeq 4.83$.}
\label{fig:SER24_214}
\end{figure}
\begin{figure}
\centering
\includegraphics[scale=1]{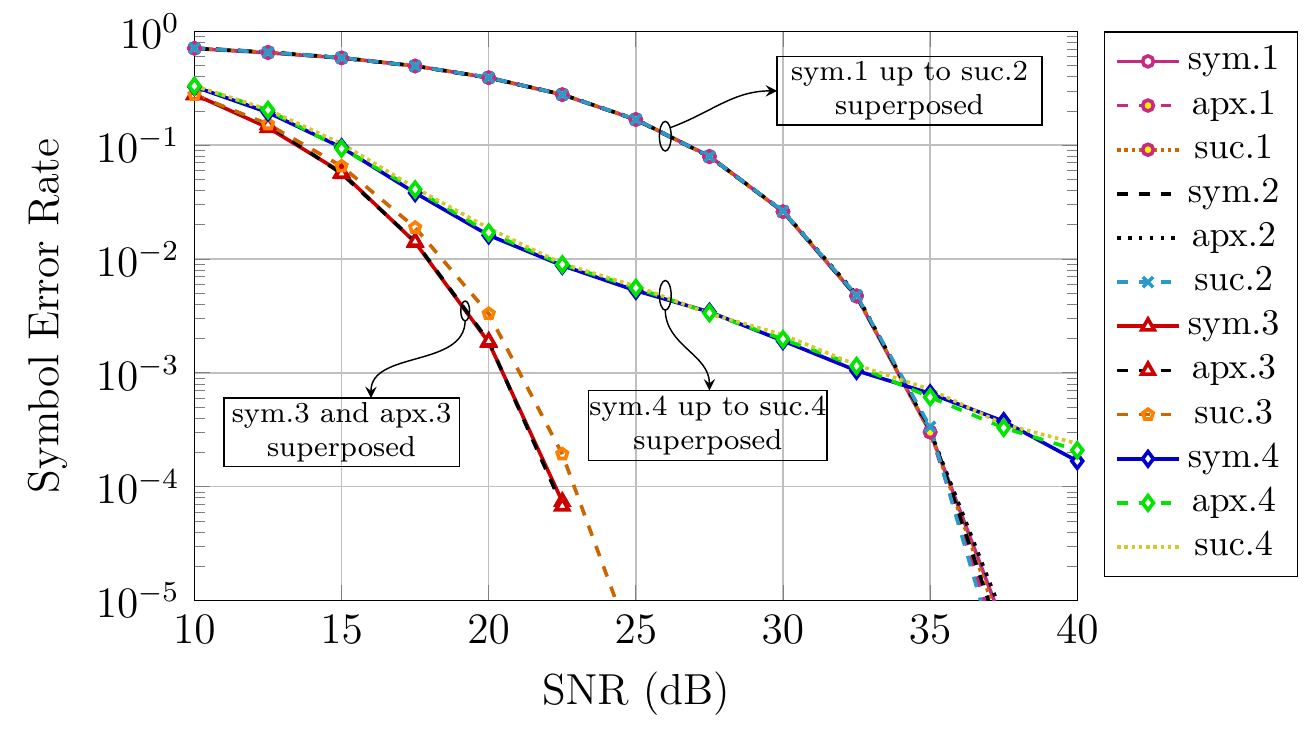}
\caption{Symbol-by-symbol detection (sym), its high-SNR approximation (apx), and successive detection (suc) of four dimensions for 8-ring/8-ary phase constellation and $\delta^2=0.69$.}
\label{fig:SER88_214}
\end{figure}
\begin{figure}
\centering
\includegraphics[scale=1]{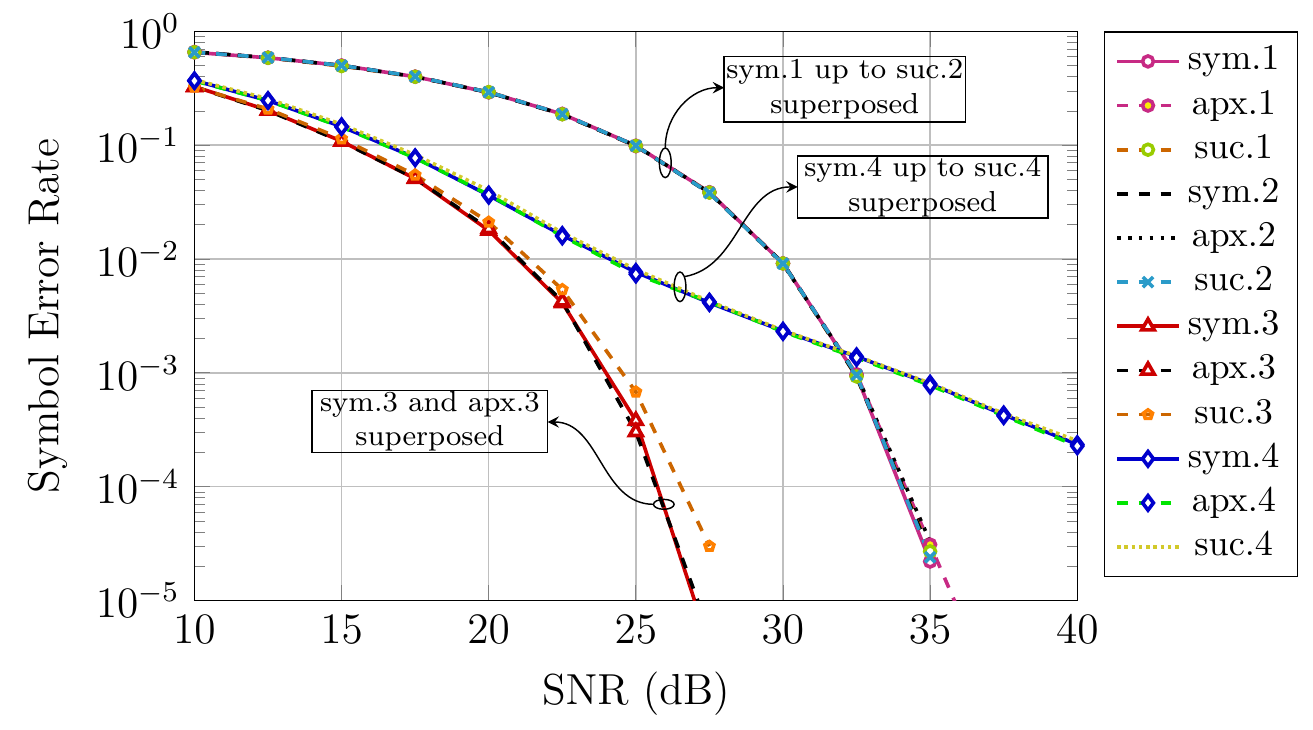}
\caption{Symbol-by-symbol detection (sym), its high-SNR approximation (apx), and successive detection (suc) of four dimensions for 8-ring/8-ary phase constellation and $\delta^2=2.12$.}
\label{fig:SER88_398}
\end{figure}
\begin{figure}
\centering
\includegraphics[scale=1]{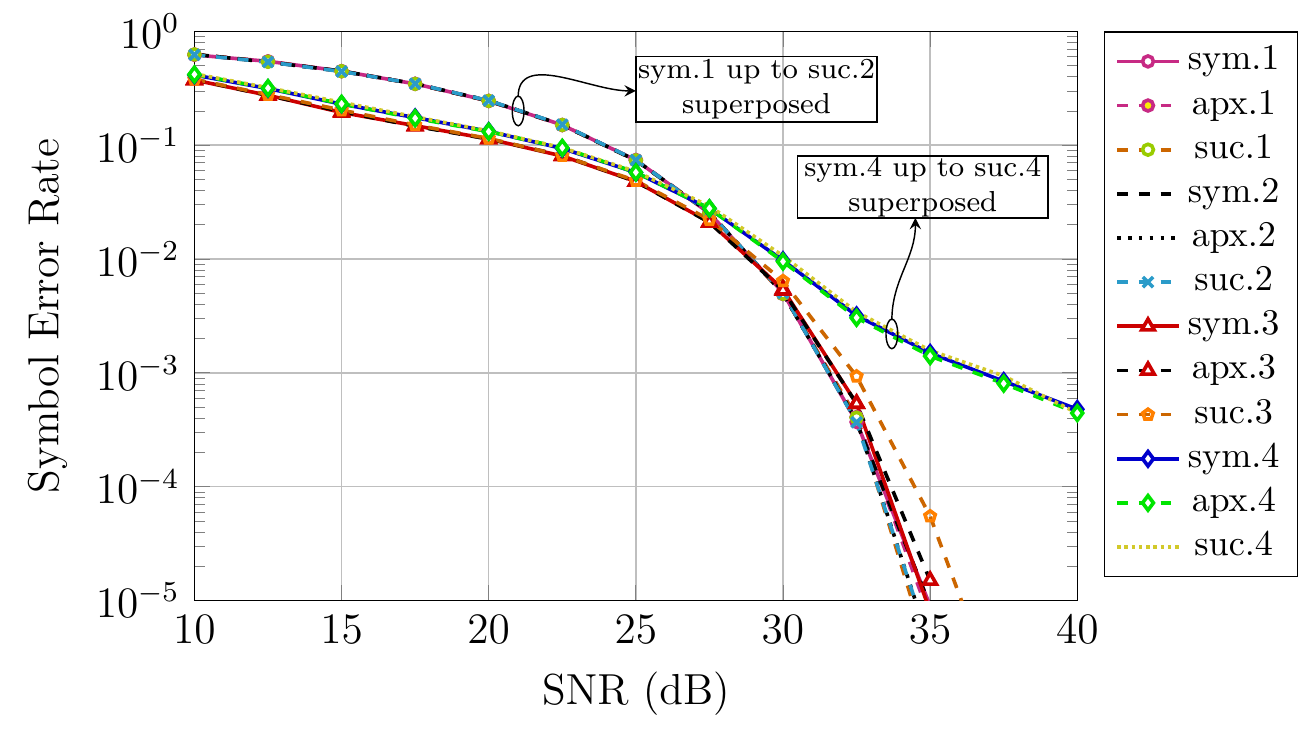}
\caption{Symbol-by-symbol detection (sym), its high-SNR approximation (apx), and successive detection (suc) of four dimensions for 8-ring/8-ary phase constellation and $\delta^2=15.36$.}
\label{fig:SER88_1042}
\end{figure}
\begin{figure}
\centering
\includegraphics[scale=1]{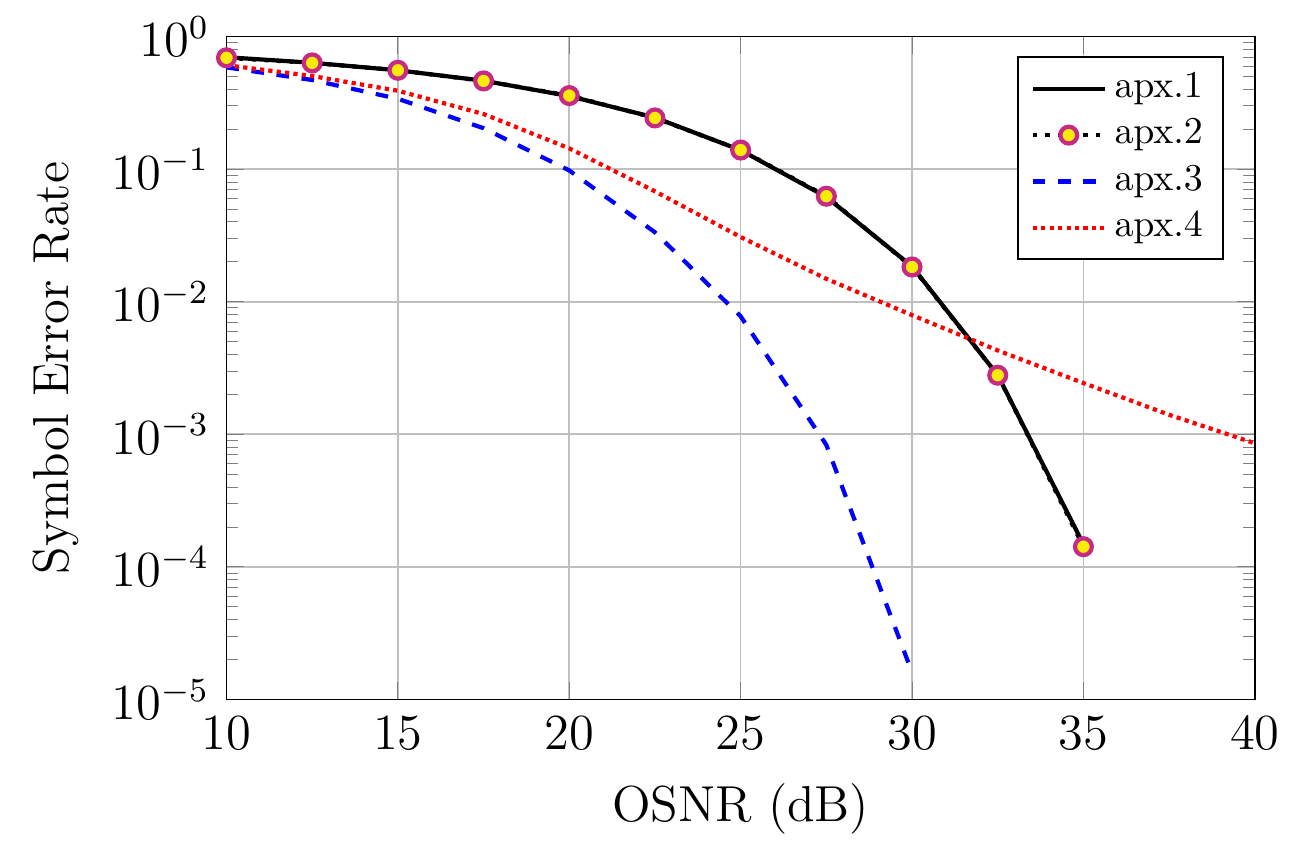}
\caption{High-SNR approximation of symbol-by-symbol detection of four dimensions for 8-ring/16-ary phase constellation and $\delta^2=0.93$.}
\label{fig:SER816_75}
\end{figure}
\begin{figure}
\centering
\includegraphics[scale=1]{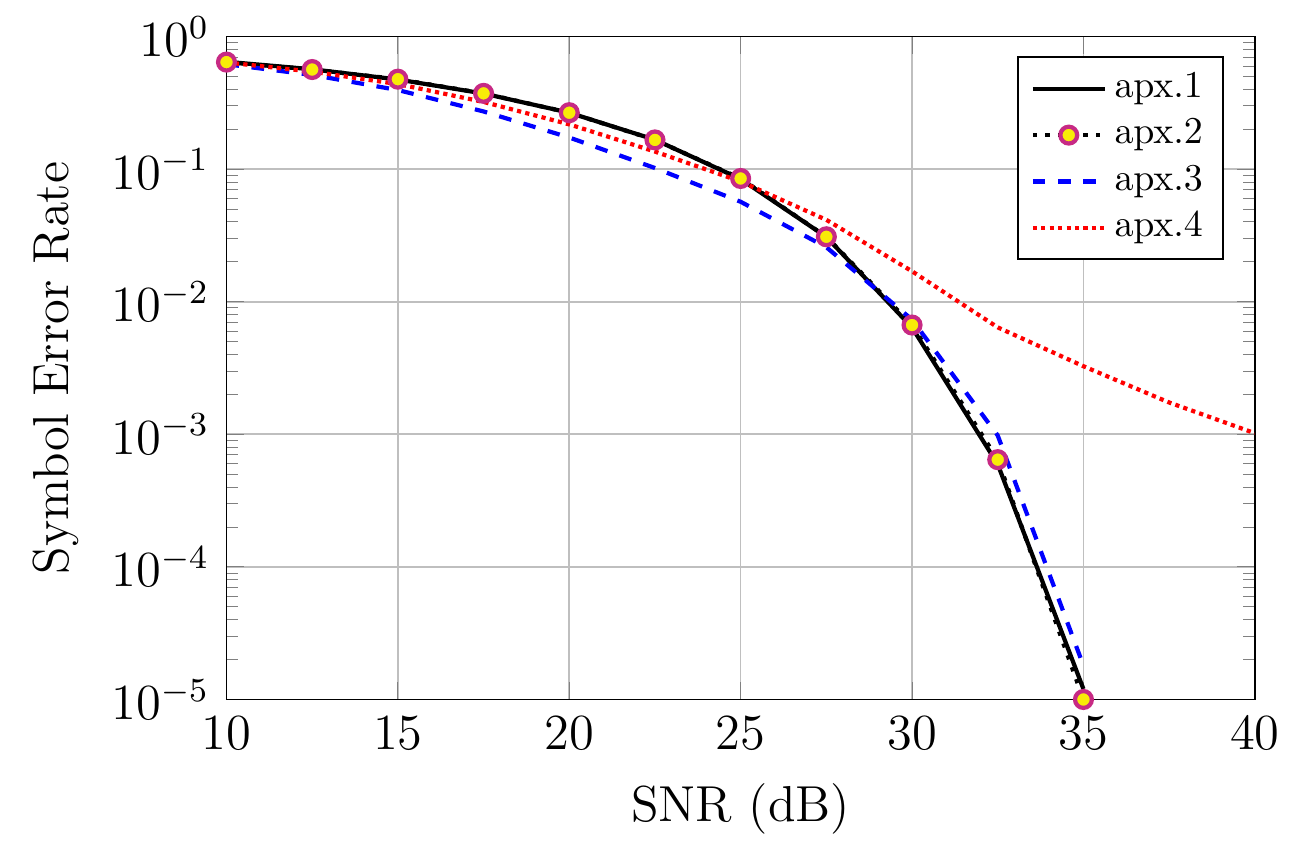}
\caption{High-SNR approximation of symbol-by-symbol detection of four dimensions for 8-ring/16-ary phase constellation and $\delta^2=4.10$.}
\label{fig:SER816_545}
\end{figure}
\begin{figure}
\centering
\includegraphics[scale=1]{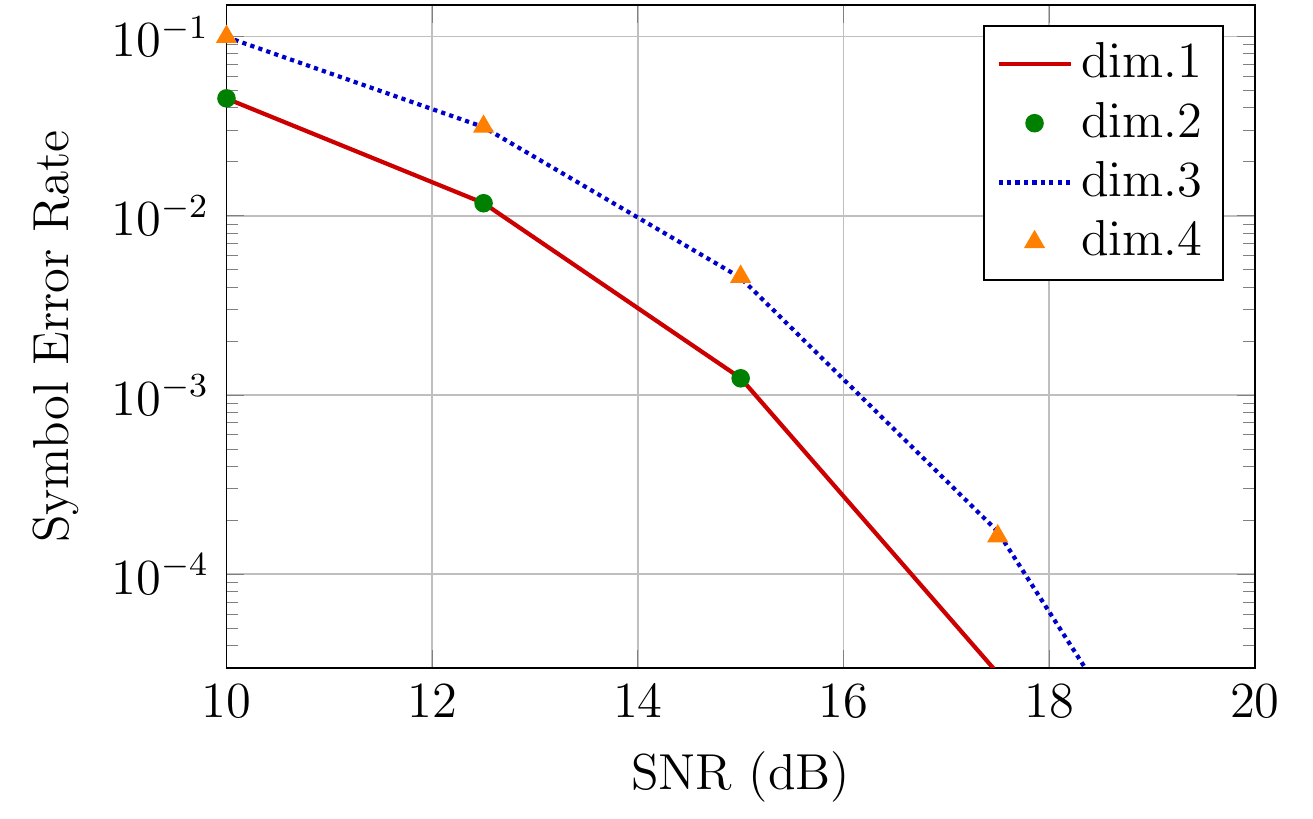}
\caption{Symbol-by-symbol detection of 2-ring/4-ary phase constellation with $\delta^2=4.10$, when the $b$ entry of $\bm{H}$ is zero (no entanglement).}
\label{fig:SER_norotate}
\end{figure}
\begin{figure}
\centering
\includegraphics[scale=1]{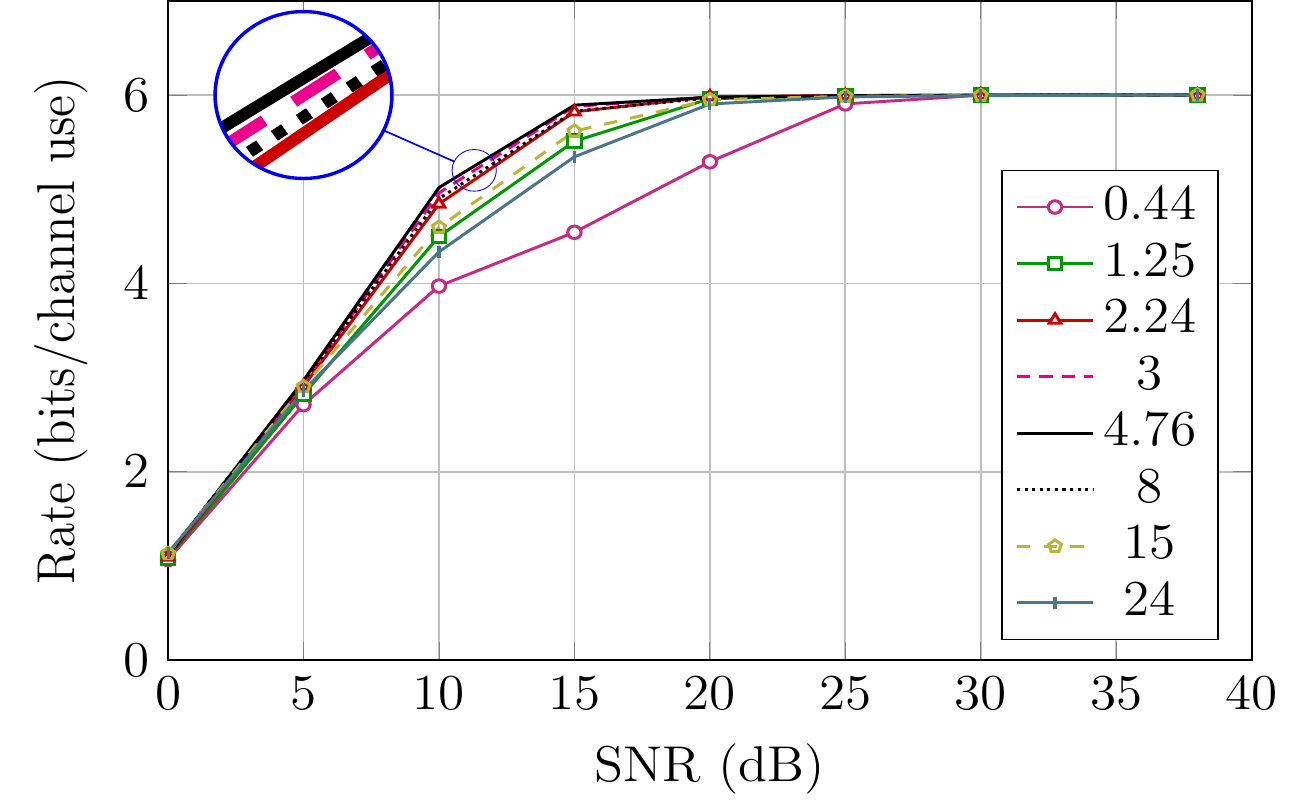}
\caption{Achievable rate of 2-ring/4-ary phase constellation for different $\delta^2$.}
\label{fig:rate24}
\end{figure}
\begin{figure}
\centering
\includegraphics[scale=1]{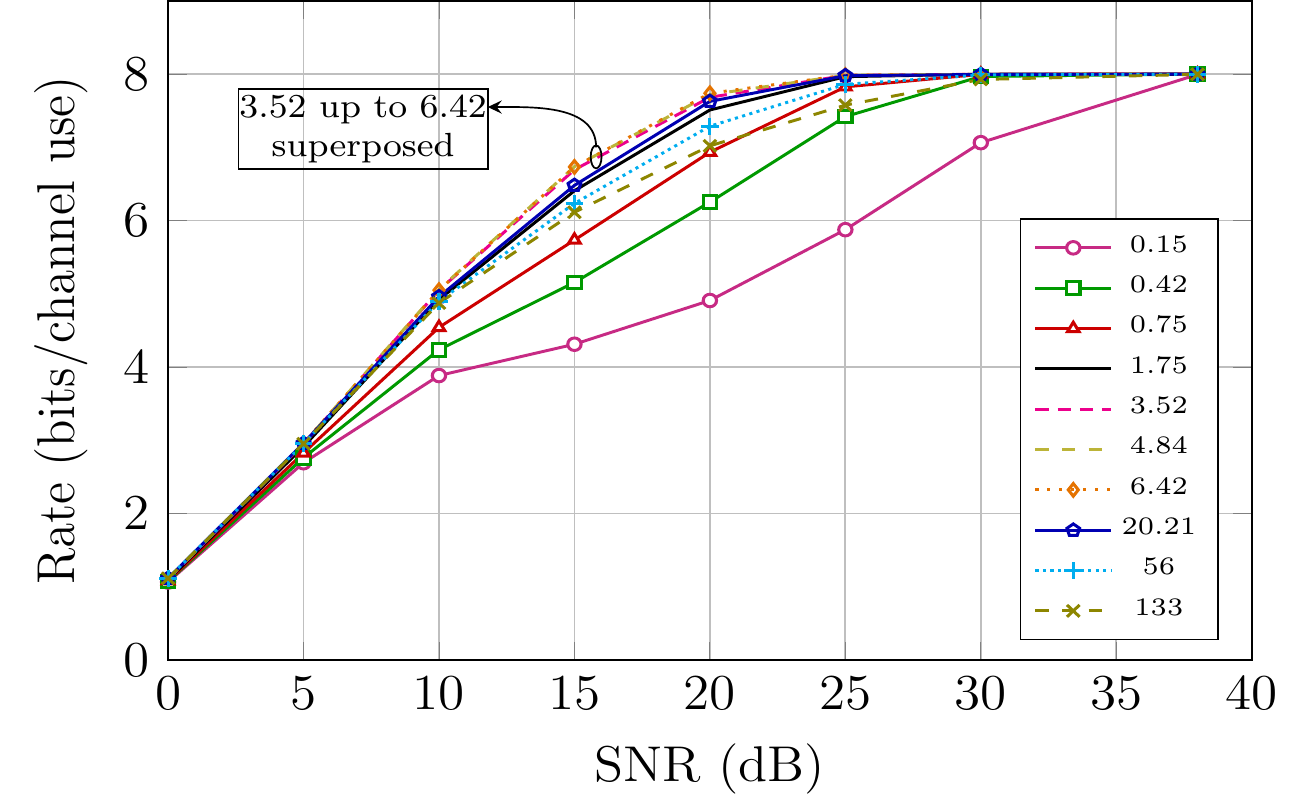}
\caption{Achievable rate of 4-ring/4-ary phase constellation for different $\delta^2$.}
\label{fig:rate44}
\end{figure}
\begin{figure}
\centering
\includegraphics[scale=1]{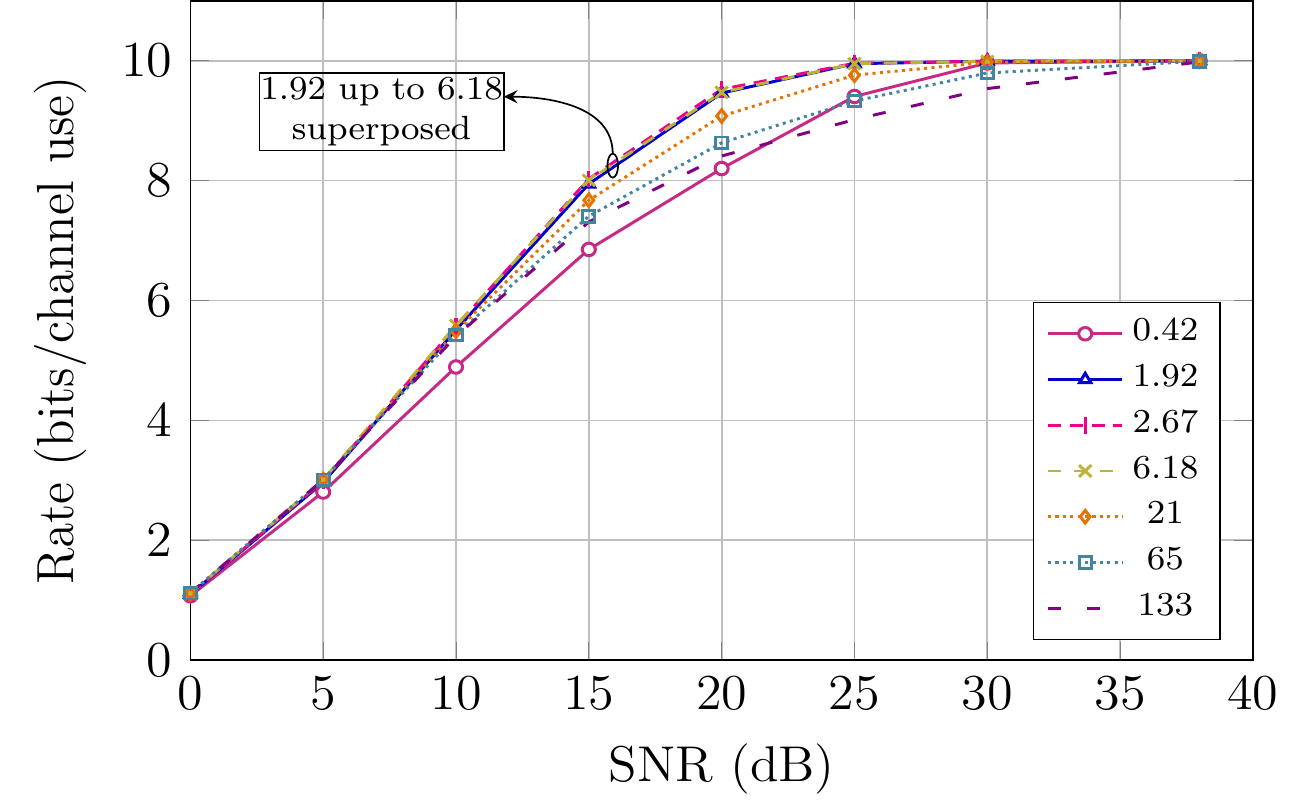}
\caption{Achievable rate of 4-ring/8-ary phase constellation for different $\delta^2$.}
\label{fig:rate48}
\end{figure}
\begin{figure}
\centering
\includegraphics[scale=1]{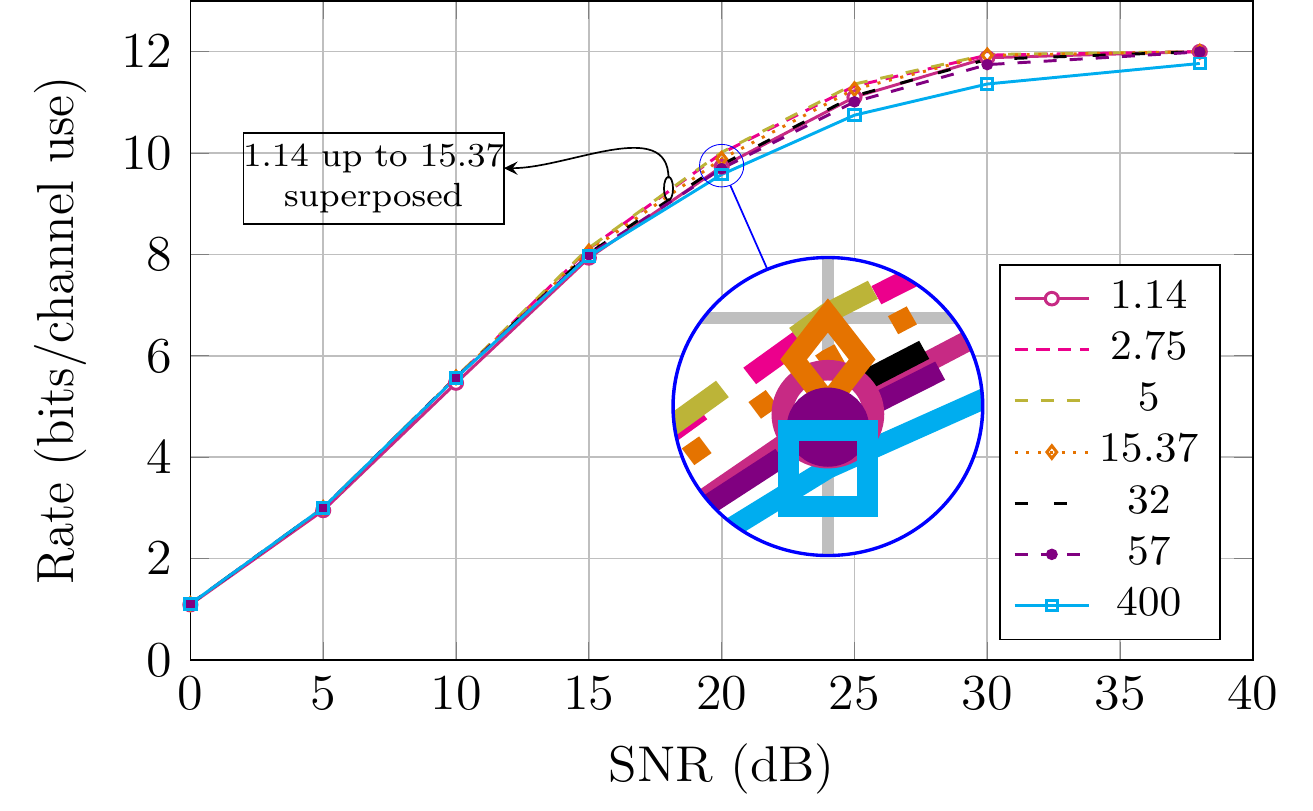}
\caption{Achievable rate of 8-ring/8-ary phase constellation for different $\delta^2$.}
\label{fig:rate88}
\end{figure}
\begin{figure}
\centering
\includegraphics[scale=1]{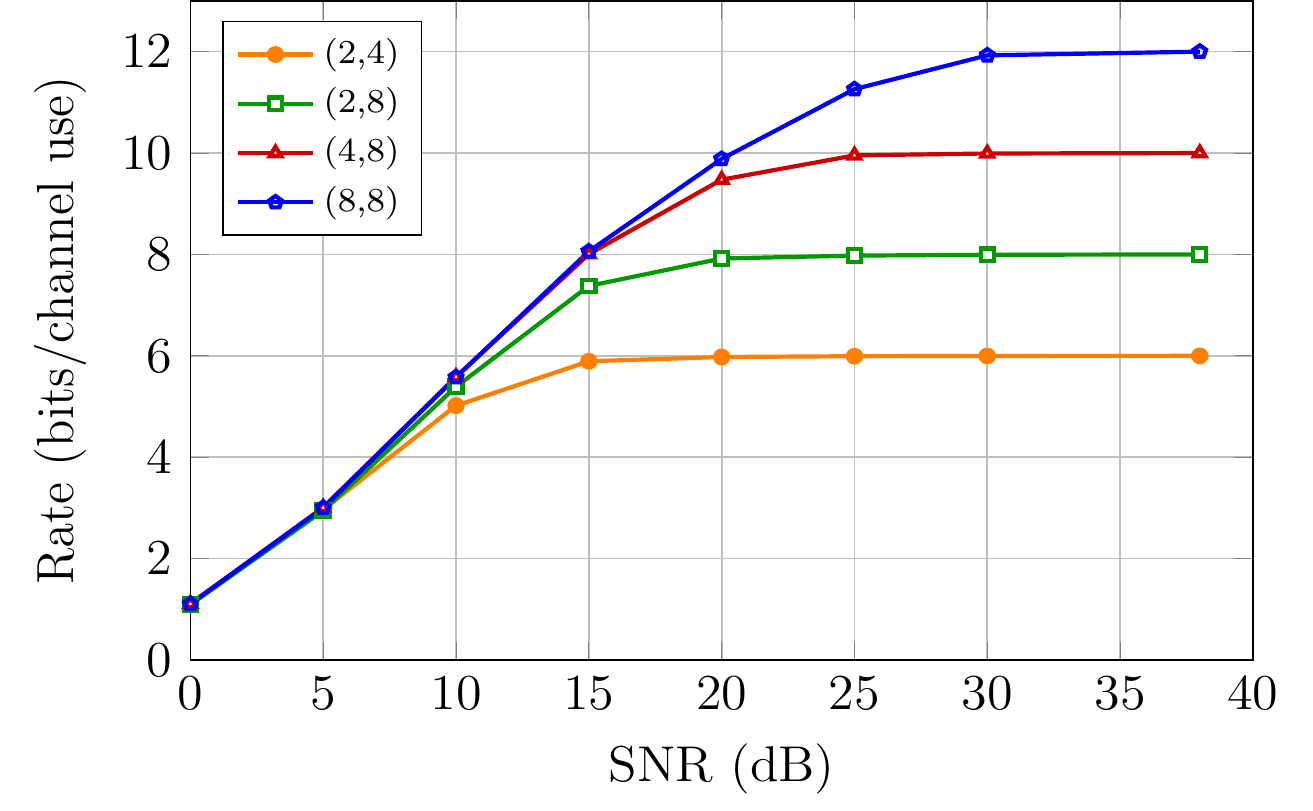}
\caption{Achievable rate of different  $n_r$-ring/$n_p$-ary phase constellations, shown as $(n_r,n_p)$, at their $\delta^2_{\text{bl}}$ (see (\ref{eq:optimal_rho}).)}
\label{fig:rate}
\end{figure}

\end{section}
\begin{section}{Conclusion}
\label{sec:Conclusion}
We computed the maximum likelihood detection rule for symbol-by-symbol and
sequence decoding, in a four-dimensional Stokes-space scheme. To reduce the
complexity of those schemes, we introduced a successive detection method. To
remedy dealing with special functions, we provided a high-SNR approximation of
the detection rules as well. We saw that the optical front-end studied in this
paper subjects the subchannel of one of the dimensions to fast fading. The
decoding methods are compared by simulations. We saw that using the successive
method results a negligible loss, and in addition, the high-SNR approximation
is very accurate (even at low SNRs). Furthermore, the achievable rates of
different constellations are obtained by the Monte Carlo method. 

An interesting future problem would be to design a a good error correcting code
and modulation, specialized for a particular application.  We have assumed that
noise contaminates the signal in the optical domain, while a more comprehensive
model might be to consider an additive noise source in the electrical domain
(after the photo diodes) as well.  In that model, a noise term must be added to
each of the six output values, $w_{1:6}$.   Determing the ML detection rule for
that scheme is left as future research. Throughout this paper, we assumed that
the receiver perfectly knows the channel matrix; it would be interesting to
conduct an analysis to determine how sensitive the detection performance is to
this assumption.

\end{section}
\begin{appendices}
\begin{section}{}
\label{app0}
\begin{proof}[Proof of Theorem~\ref{thrm:phase}] The phase of $R_u$, $\Psi_u$, has a von Mises distribution with the PDF given as~\cite{OnTheProduct}
\[f(\psi_u~\mid~ |r_u|,|k_u|)=\frac{\exp(\lambda_u\cos(\psi_u-\phi_u))}{2\pi\mathcal{I}_0(\lambda_u)}.\]
Note that
\begin{align}
\begin{array}{cc}
\theta' = \phi_x-\phi_y,&\quad\gamma'= \phi_x-\D\phi_y,\\
\theta''=\psi_x-\psi_y,&\quad\gamma''=\psi_x-\D\psi_y,
\end{array}
\label{eq:difference}
\end{align}
and as a result, $\theta''$ and $\gamma''$ are functions of $\psi_x, \psi_y,$ and $\D\psi_y$. Therefore, we use the Jacobian of this transformation to compute the joint conditional PDF of $\Theta''$ and $\Gamma''$ from the joint conditional PDF of $\Psi_x, \Psi_y,$ and $\D\Psi_y$~\cite[p.~244]{papoulis}. To use the Jacobian, the number of random variables before and after the transformation must be the same, which is not true in this case. We introduce a dummy random variable, $\Omega=\Psi_x$, and find the joint conditional PDF of $\Theta'', \Gamma''$, and $\Omega$. Then by marginalizing $\Omega$, we obtain the PDF of interest. 

The determinant of the Jacobian matrix is 
\[\left|\frac{\partial(\Theta'',\Gamma'',\Omega)}{\partial(\Psi_x,\Psi_y,\D\Psi_{y})}\right|=\left|\left[\begin{array}{ccc}
1 & -1 & 0\\
1 & 0 & -1\\
1 & 0 & 0
\end{array}\right]\right|=1,\]
where, {\it e.g.,} the element in the first row and the second column of the Jacobian matrix is $\frac{\partial\Theta''}{\partial\Psi_y}=-1$. As there is a one-to-one correspondence between $(\Theta'',\Gamma'',\Omega)$ and $(\Psi_x,\Psi_y,\D\Psi_{y})$, there is only a unique $(\psi_x,\psi_y,\D\psi_y)$ that contributes to the PDF of $(\theta'',\gamma'',\omega)$. Let $\bm{c}\triangleq[|r_x|,|r_y|,|\D r_y|,\bm{d}_k]^t$ denote the condition vector. Then we have
\begin{align*}
f_{\Theta'',\Gamma'',\Omega\mid \bm{c}}(\theta'',\gamma'',\omega\mid \bm{c})
&=f_{\Psi_x,\Psi_y,\D\Psi_y\mid \bm{c}}(\omega,\omega-\theta'',\omega-\gamma''\mid \bm{c})\\
&\stackrel{(i)}{=}\frac{e^{\lambda_x\cos(\omega-\phi_x)+\lambda_y\cos(\omega-\theta''-\phi_y)+\D\lambda_y\cos(\omega-\gamma''-\D\phi_y)}}{8\pi^3\mathcal{I}_0(\lambda_x)\mathcal{I}_0(\lambda_y)\mathcal{I}_0(\D\lambda_y)}\\
&\stackrel{(ii)}{=}\frac{\exp\left(\frac{|\langle\bm{d}_k,\bm{d}_r\rangle|}{\sigma^2}\cos(\omega-\phi_x+\alpha+\beta)\right)}{8\pi^3\mathcal{I}_0(\lambda_x)\mathcal{I}_0(\lambda_y)\mathcal{I}_0(\D\lambda_y)}.
\end{align*}
where $(i)$ is due to the independence of $\Psi_x,\Psi_y$ and $\D\Psi_y$, and $(ii)$ is true as, by using Fig.~\ref{fig:triangles}, we have
\begin{align*}
\lambda_x e^{i(\omega-\phi_x)}+\lambda_ye^{i(\omega-\theta''-\phi_y)}
&+\D\lambda_ye^{i(\omega-\gamma''-\D\phi_y)}\\
&=
\left(|k_x|\cdot|r_x|+|k_y|\cdot|r_y|e^{i(\theta'-\theta'')}+|\D k_y|\cdot|\D r_y|e^{i(\gamma'-\gamma'')}\right)\cdot
\frac{e^{i(\omega-\phi_x)}}{\sigma^2}\\
&=\frac{|\langle\bm{d}_k,\bm{d}_r\rangle|}{\sigma^2} \exp\left(i(\omega-\phi_x+\alpha+\beta)\right).
\end{align*}
As a result, 
\begin{align*}
\lambda_x\cos(\omega-\phi_x)&+\lambda_y\cos(\omega-\theta''-\phi_y)
+\D\lambda_y\cos(\omega-\gamma''-\D\phi_y)\\
&=\Re\left(\lambda_xe^{i(\omega-\phi_x)}+\lambda_ye^{i(\omega-\theta''-\phi_y)}+\D\lambda_ye^{i(\omega-\gamma''-\D\phi_y)}\right)\\
&=\frac{|\langle\bm{d}_k,\bm{d}_r\rangle|}{\sigma^2}\cos(\omega-\phi_x+\alpha+\beta).
\end{align*}

By marginalizing $\Omega$, we have
\begin{align*}
f_{\Theta'',\Gamma''\mid \bm{c}}(\theta'',\gamma''\mid \bm{c})
&=\int_{0}^{2\pi}\frac{\exp\left(\frac{|\langle\bm{d}_k,\bm{d}_r\rangle|}{\sigma^2}\cos(\omega-\phi_x+\alpha+\beta)\right)}{8\pi^3\mathcal{I}_0(\lambda_x)\mathcal{I}_0(\lambda_y)\mathcal{I}_0(\D\lambda_y)}\dd\omega\\
&\stackrel{(i)}{=}\frac{\int_{0}^{2\pi}\exp\left(\frac{|\langle\bm{d}_k,\bm{d}_r\rangle|}{\sigma^2}\cos(\omega)\right)\dd\omega}{8\pi^3\mathcal{I}_0(\lambda_x)\mathcal{I}_0(\lambda_y)\mathcal{I}_0(\D\lambda_y)}\\
&=\frac{\mathcal{I}_0(\frac{|\langle\bm{d}_k,\bm{d}_r\rangle|}{\sigma^2})}{4\pi^2\mathcal{I}_0(\lambda_x)\mathcal{I}_0(\lambda_y)\mathcal{I}_0(\D\lambda_y)},
\end{align*}
where in $(i)$, we have ignored the constant phase offset, $-\phi_x+\alpha+\beta$, as we are integrating over one period of cosine function. 
\end{proof}

\end{section}
\begin{section}{}
\label{app1}
\begin{proof}[Proof of Theorem~\ref{thrm:suc1}]
We adopt the proof given in~\cite{OnTheProduct}, with some adaptation to be consistent with our notation. Let $\bm{g}=[\hat{\bm{d}}_k,|r_x|,|r_y|]^t$. By the definition of conditional PDF, we have
\[f(|r_x|,|r_y|,\theta''\mid\hat{\bm{d}}_k)=f(|r_x|,|r_y|\mid\hat{\bm{d}}_k) f(\theta''\mid\bm{g}),\]
so, to find $f(|r_x|,|r_y|,\theta''\mid\hat{\bm{d}}_k)$, we compute $f(|r_x|,|r_y|\mid\hat{\bm{d}}_k)$ and $f(\theta''\mid\bm{g})$.

As said in the proof of Theorem~\ref{thrm:radii}, given $|K_u|$, $|R_u|$ becomes independent of $\Theta'$, for $u\in\{x,y\}$. As a result, according to Theorem~\ref{thrm:radii}, we have
\begin{equation}
f(|r_x|,|r_y|\mid \hat{\bm{d}}_k))
=\frac{|r_x|\cdot|r_y|}{\sigma^4}\exp\left(\frac{-(|\hat{\bm{d}}_k|^2+|\hat{\bm{d}}_r|^2)}{2\sigma^2}\right)
\cdot\mathcal{I}_0(\lambda_x)\mathcal{I}_0(\lambda_y).
\label{eq:r}
\end{equation}
To compute $f(\theta''\mid\bm{g})$, note that $\Theta''$ is the subtraction of two independent von Mises random variables (see~(\ref{eq:difference}).) As a result, its PDF is the convolution of two von Mises PDFs, given as
\begin{align}
f(\theta''\mid\bm{g})&=\int_{0}^{2\pi}f_{\Psi_1}(\psi_1)f_{\Psi_2}(\psi_1-\theta'')\dd\psi_1\nonumber\\
&=\frac{\int_{0}^{2\pi}e^{\lambda_x\cos(\psi_1-\phi_1)+\lambda_y\cos(\psi_1-\theta''-\phi_2)}\dd\psi_1}{4\pi^2\mathcal{I}_0(\lambda_x)\mathcal{I}_0(\lambda_y)}\nonumber\\
&\stackrel{(i)}{=}\frac{\int_{0}^{2\pi}\exp\left(\frac{|\langle\hat{\bm{d}}_k,\hat{\bm{d}}_r\rangle|}{\sigma^2}\cos(\psi_1-\phi_1+\alpha)\right)\dd\psi_1}{4\pi^2\mathcal{I}_0(\lambda_x)\mathcal{I}_0(\lambda_y)}\nonumber\\
&=\frac{\mathcal{I}_0\left(\frac{|\langle\hat{\bm{d}}_k,\hat{\bm{d}}_r\rangle|}{\sigma^2}\right)}{2\pi\mathcal{I}_0(\lambda_x)\mathcal{I}_0(\lambda_y)},
\label{eq:angle}
\end{align}
where $(i)$ is true as
\begin{align*}
\lambda_x\cos(\psi_1-\phi_1)+\lambda_y\cos(\psi_1-\theta''-\phi_2)
&=\Re\left(\lambda_xe^{i(\psi_1-\phi_1)}+\lambda_ye^{i(\psi_1-\theta''-\phi_2)}\right)\\
&=\Re\left(\frac{e^{i(\psi_1-\phi_1)}}{\sigma^2}\left(|k_x|\cdot|r_x|+|k_y|\cdot|r_y|e^{i(\theta'-\theta'')}\right)\right)\\
&\stackrel{(ii)}{=}\Re\left(\frac{e^{i(\psi_1-\phi_1)}}{\sigma^2}|\langle\hat{\bm{d}}_k,\hat{\bm{d}}_r\rangle|e^{i\alpha}\right)\\
&=\frac{|\langle\hat{\bm{d}}_k,\hat{\bm{d}}_r\rangle|}{\sigma^2}\cos(\psi_1-\phi_1+\alpha),
\end{align*}
and for $(ii)$ see Fig.~\ref{fig:triangles}~\cite[p.~44]{mardia2009directional}. Using~(\ref{eq:r}) and~(\ref{eq:angle}), we have
\begin{align*}
 f(|r_x|,|r_y|,\theta''\mid\hat{\bm{d}}_k)=
\frac{|r_x|\cdot|r_y|}{2\pi\sigma^4}\exp\left(\frac{-\left(|\hat{\bm{d}}_r|^2+|\hat{\bm{d}}_k|^2\right)}{2\sigma^2}\right)\mathcal{I}_0\left(\frac{|\langle\hat{\bm{d}}_k,\hat{\bm{d}}_r\rangle|}{\sigma^2}\right).
\end{align*}
\end{proof}

\end{section}
\begin{section}{}
\label{app2}
\begin{proof}[Proof of Theorem~\ref{thrm:fourth}]
Given $\hat{\bm{d}}_k$, $\hat{\bm{d}}_r$ becomes independent of $(\gamma',|\D r_y|,|\D k_y|)$. As a result, we have 
\begin{align*}
\underset{\gamma'}{\arg\max}~f(\gamma''\mid\bm{d}_k,|\D r_y|,|r_x|,|r_y|,\theta'')
&=\underset{\gamma'}{\arg\max}~\frac{f(|r_x|,|r_y|,\theta'',\gamma''\mid \bm{d}_k,|\D r_y|)}{f(|r_x|,|r_y|,\theta''\mid\bm{d}_k,|\D r_y|)}\\
&=\underset{\gamma'}{\arg\max}~\frac{f(|r_x|,|r_y|,\theta'',\gamma''\mid \bm{d}_k,|\D r_y|)}{f(|r_x|,|r_y|,\theta''\mid |k_x|,|k_y|,\theta')}\\
&\stackrel{(i)}{=}\underset{\gamma'}{\arg\max}~\frac{\mathcal{I}_0\left(\frac{|\langle\bm{d}_k,\bm{d}_r\rangle|}{\sigma^2}\right)}{\mathcal{I}_0\left(\frac{|\langle\hat{\bm{d}}_k,\hat{\bm{d}}_r\rangle|}{\sigma^2}\right)}\\
&\stackrel{(ii)}{=}\underset{\gamma'}{\arg\max}~|\langle\bm{d}_k,\bm{d}_r\rangle|\\
&=\underset{\gamma'}{\arg\max}~\cos(\gamma'-\gamma''-\alpha),\\
\end{align*}
where $(i)$ is true due to~(\ref{eq:4d}) and Theorem~\ref{thrm:suc1}, and $(ii)$ is true as, at this step, $|\langle\hat{\bm{d}}_k,\hat{\bm{d}}_r\rangle|$ is the same for all feasible $\gamma'$. 
\end{proof}

\end{section}
\end{appendices}
\bibliographystyle{IEEEtran}
\bibliography{main}
\end{document}